\documentclass[sigconf,nonacm]{acmart}

\AtBeginDocument{%
  }

\usepackage{amsmath,amsthm}

\usepackage{enumitem}
\usepackage[table]{xcolor}
\usepackage{graphicx}
\usepackage{subcaption}
\usepackage{booktabs}
\usepackage{multirow}
\usepackage{makecell}
\usepackage{tabularx}
\usepackage{array}
\usepackage{xspace}
\usepackage{float}
\usepackage{balance}
\usepackage{hyperref}
\usepackage[normalem]{ulem}
\usepackage{soul}
\usepackage{pifont}
\usepackage{algorithm}
\usepackage{algpseudocode}
\usepackage[sets]{cryptocode}
\usepackage{listings}
\usepackage{upquote}
\usepackage[most]{tcolorbox}
\tcbuselibrary{skins,breakable}
\usepackage{tikz}
\usetikzlibrary{arrows.meta,positioning,shapes.geometric,calc}

\sethlcolor{gray!20}
\theoremstyle{definition}
\newtheorem{definition}{Definition}
\theoremstyle{plain}
\newtheorem{theorem}{Theorem}
\newtheorem{lemma}{Lemma}
\newtheorem{proposition}{Proposition}

\newcommand{\system}{\textsc{Uber-Join}\xspace}
\newcommand{\joinhard}{\textsc{Join-Hard}\xspace}

\newcommand{\stitle}[1]{\noindent\textbf{#1}\xspace}

\newcommand{\squishlist}{%
  \begin{list}{$\bullet$}{\leftmargin=1em}%
}
\newcommand{\squishend}{\end{list}}

\newcommand*\circled[1]{\tikz[baseline=(char.base)]{%
  \node[shape=circle,draw,inner sep=0.5pt] (char) {#1};}}

\definecolor{rankfirst}{HTML}{C8DEC3}
\definecolor{ranksecond}{HTML}{FBEDC3}
\definecolor{rankthird}{HTML}{F2C6B6}
\newcommand{\rkfirst}[1]{{\setlength{\fboxsep}{1pt}\colorbox{rankfirst}{#1}}}
\newcommand{\rksecond}[1]{{\setlength{\fboxsep}{1pt}\colorbox{ranksecond}{#1}}}
\newcommand{\rkthird}[1]{{\setlength{\fboxsep}{1pt}\colorbox{rankthird}{#1}}}
\newcommand{\sysrowstrut}{\rule[-\dimexpr\dp\strutbox+\aboverulesep\relax]{0pt}{\dimexpr\ht\strutbox+\dp\strutbox+\aboverulesep+\belowrulesep\relax}}

\newtheoremstyle{exampleplain}{.5\baselineskip}{.5\baselineskip}{\normalfont}{\parindent}{\scshape}{.}{.5em}{\thmname{#1}\thmnumber{ #2}}
\theoremstyle{exampleplain}

\definecolor{exAccent}{HTML}{5B8CC8}
\definecolor{exBack}{HTML}{EBF2F9}
\newtcolorbox{exampleaccent}{enhanced jigsaw, breakable, parbox=false, frame hidden, boxrule=0pt, arc=0pt,
  colback=exBack, borderline west={2.2pt}{0pt}{exAccent}, left=8pt, right=8pt, top=7pt, bottom=7pt,
  boxsep=0pt, before skip=6pt, after skip=6pt}
\newcommand{\example}[1]{%
  \begin{exampleaccent}\noindent\refstepcounter{exampleenv}\textbf{Example \theexampleenv.}\hspace{.4em}#1\end{exampleaccent}%
}

\begin{document}
\title{\system: Unifying Diverse Table Joins using  LLMs with Precision and Cost Guarantees}

\author{Guohui Guan}
\affiliation{%
  \institution{University of Minnesota}
  \country{}}
\email{gguan@umn.edu}

\author{Chang Ge}
\affiliation{%
  \institution{University of Minnesota}
  \country{}}
\email{cge@umn.edu}

\author{Yeye He}
\affiliation{%
  \institution{Microsoft Research}
  \country{}}
\email{yeyehe@microsoft.com}

\begin{abstract}

When two tables from different sources are joined, the values that should join often differ in spelling, format, or meaning, and fuzzy-, transformation-, and semantic-joins each handle one type of these differences.
Large language models (LLMs) can match values that differ in several ways at once, but every LLM call is billed, and having an LLM compare the rows of two large tables can cost far more than a typical budget allows.
This paper studies how to combine these join types with LLMs to join two tables under a dollar budget, maximizing recall while keeping the precision of the result at or above a target, a problem that is \textsf{NP}-hard.

Our key idea is to use LLMs selectively to create join programs and verify their outputs through sampling, while adaptively deciding which joins to pursue under budget and precision constraints.
This combines the scalability of program-based joins with selective LLM verification, while maintaining a statistical precision guarantee on the accumulated result.
We develop \system{}, which treats creating a join and checking the output of a created join as two levels of arms in a budgeted multi-armed bandit.
It commits an output only if the precision of the whole result, estimated from an LLM's judgments of sampled pairs, stays at or above the target with high confidence.
Under the same budget, \system{} achieves the highest recall at the precision target in nearly all settings on three join benchmarks and \joinhard{}, a new benchmark of large tables.

\end{abstract}

\maketitle

\section{Introduction}
\label{sec:intro}

\begin{figure*}[t]
\centering
\footnotesize
\def\fz#1{\textbf{\textcolor{red!70!black}{#1}}}
\def\nm#1{\uline{#1}}
\def\ct#1{\dashuline{#1}}
\def\ab#1{\textbf{\textcolor{blue!60!black}{#1}}}
\begin{tikzpicture}[
  hospitalcell/.style={anchor=north west, align=left,
               text width=\dimexpr#1\textwidth-4pt\relax,
               inner xsep=2pt, inner ysep=3pt, outer sep=0pt},
  note/.style={anchor=west, align=left, text width=0.25\textwidth,
               inner sep=0pt, font=\footnotesize},
  rule/.style={draw=black!45, line width=0.35pt}]
\node[anchor=south, font=\footnotesize\bfseries] at (0.18\textwidth,0.04)
  {Table $S$};
\node[anchor=south, font=\footnotesize\bfseries] at (0.535\textwidth,0.04)
  {Table $T$};
\coordinate (r0) at (0,-17pt);
\foreach \left/\divider/\right/\id in {0/0.035/0.36/S,0.39/0.425/0.68/T} {
  \fill[black!6] (\left\textwidth,0) rectangle (r0 -| \right\textwidth,0);
  \node[font=\scriptsize\bfseries, inner sep=0pt]
    at ({(\left+\divider)/2*\textwidth},-8.5pt) {\id-ID};
  \node[font=\scriptsize\bfseries, inner sep=0pt]
    at ({(\divider+\right)/2*\textwidth},-8.5pt) {Hospital Name};
}
\foreach \i/\sname/\tname [evaluate=\i as \previous using int(\i-1)] in {
  1/{\fz{The} Hospital for Sick Children}/{Hospital for Sick Children},
  2/{Mease Dunedin Hosp\fz{ti}al}/{Mease Dunedin Hosp\fz{it}al},
  3/{Korle\fz{-}Bu Teaching Hospital}/{Korle Bu Teaching Hospital},
  4/{102 Pok Fu Lam Road, \ct{Hong Kong}\\/ \nm{Queen Mary Hospital}}/{\nm{Queen Mary Hospital} (\ct{Hong Kong})},
  5/{name: \nm{Prince of Wales Hospital}; street: Barker Street; city: \ct{Sydney}; country: Australia}/{\nm{Prince of Wales Hospital}, \ct{Sydney}},
  6/{55 Fruit Street, \ct{Boston}\\/ \nm{Massachusetts General Hospital}}/{\nm{Massachusetts General Hospital} (\ct{Boston})},
  7/{\ab{CHOP}}/{\ab{C}hildren's \ab{H}ospital \ab{o}f \ab{P}hiladelphia},
  8/{\ab{H}ospital for \ab{S}pecial \ab{S}urgery}/{\ab{HSS}},
  9/{Universit\"atsklinikum Heidelberg}/{Heidelberg University Hospital},
  10/{30 Gascoigne Road, \ct{Kowloon}, Hong Kong SAR, China\\/ \nm{Queen Elizabeth Hosp\fz{ti}al}}/{\nm{Queen Elizabeth Hosp\fz{it}al} (\ct{Kowloon})},
  11/{1 Gustave L. Levy Place, \ct{New York}, NY, USA\\/ \nm{\fz{Mt.} Sinai Hospital}}/{\nm{\fz{Mount} Sinai Hospital} (\ct{New York})}} {
  \node[hospitalcell=0.325] (s\i) at (r\previous -| 0.035\textwidth,0) {\sname};
  \node[hospitalcell=0.255] (t\i) at (r\previous -| 0.425\textwidth,0) {\tname};
  \path let \p1=(s\i.south west), \p2=(t\i.south west) in
    coordinate (r\i) at (0,{min(\y1,\y2)});
  \coordinate (c\i) at ($(r\previous)!0.5!(r\i)$);
  \node[inner sep=0pt] at (c\i -| 0.0175\textwidth,0) {$s_{\i}$};
  \node[inner sep=0pt] at (c\i -| 0.4075\textwidth,0) {$t_{\i}$};
  \draw[green!50!black, line width=0.45pt, {Stealth[length=2.4pt]}-{Stealth[length=2.4pt]}]
    (c\i -| 0.364\textwidth,0) -- (c\i -| 0.386\textwidth,0);
}
\foreach \left/\divider/\right in {0/0.035/0.36,0.39/0.425/0.68} {
  \draw[line width=0.65pt] (\left\textwidth,0)
    rectangle (r11 -| \right\textwidth,0);
  \draw[rule] (\divider\textwidth,0) -- (r11 -| \divider\textwidth,0);
  \foreach \i in {0,...,10} {
    \draw[rule] (r\i -| \left\textwidth,0) -- (r\i -| \right\textwidth,0);
  }
  \foreach \i in {3,6,9} {
    \draw[black!75, line width=0.65pt]
      (r\i -| \left\textwidth,0) -- (r\i -| \right\textwidth,0);
  }
}
\foreach \upper/\lower/\name in {r0/r3/f,r3/r6/x,r6/r9/s,r9/r11/c} {
  \coordinate (n\name) at ($(\upper)!0.5!(\lower)$);
  \draw[black!50, line width=0.45pt]
    ($(\upper -| 0.828\textwidth,0)+(0,-3pt)$) --
    ($(\upper -| 0.84\textwidth,0)+(0,-3pt)$) --
    ($(\lower -| 0.84\textwidth,0)+(0,3pt)$) --
    ($(\lower -| 0.828\textwidth,0)+(0,3pt)$);
  \draw[black!50, line width=0.45pt]
    (n\name -| 0.84\textwidth,0) -- (n\name -| 0.85\textwidth,0);
}
\foreach \i/\v in {1/extra word, 2/typo, 3/punctuation, 4/{reorder, drop street}, 5/{key-value to list, drop fields}, 6/{reorder, drop street}, 7/{abbreviation in $S$}, 8/{abbreviation in $T$}, 9/translation, 10/{reorder, drop address + typo}, 11/{reorder, drop address + short form}} {
  \node[anchor=west, font=\small\itshape, text=black!70, inner sep=0pt, text width=0.13\textwidth, align=left, execute at begin node={\hyphenpenalty=10000\relax}] at (c\i -| 0.69\textwidth,0) {\v};
}
\foreach \n/\g in {f/{(J1)\\fuzzy-join}, x/{(J2)\\transformation-join}, s/{(J3)\\semantic-join}, c/{(J2 + J1)\\composition}} {
  \node[anchor=west, font=\footnotesize\bfseries, inner sep=0pt, align=left] at (n\n -| 0.855\textwidth,0) {\g};
}
\end{tikzpicture}
\caption{Sample rows of two hospital tables (Example~\ref{ex:hospital}); each $(s_i,t_i)$ is a matching pair. \protect\fz{Red} marks small edits, \protect\nm{solid} and \protect\ct{dashed} underlines mark rearranged fields, and \protect\ab{blue} marks abbreviation letters.}
\label{fig:running}
\end{figure*}

As organizations accumulate data across systems, answering a question often involves combining information from multiple tables.
For example, an analyst may combine a sales table with a product catalog maintained in another system to report the revenue of each product category.
In relational databases, this combination is performed through joins.
A join identifies pairs of rows whose join keys satisfy a specified predicate.
One of the most common cases is an equi-join, where the predicate requires the corresponding key values to be exactly equal.

Exact equality, however, can be too restrictive when tables come from different sources.
Values that should correspond may differ in spelling, formatting, representation, or terminology, causing correct row pairs to remain unmatched~\cite{autojoin, semajoin}.
Therefore, techniques have been proposed to automatically handle diverse types of joins beyond exact equality, such as \underline{\textbf{(J1)} fuzzy-join} that handles small string-level differences (e.g., Auto-FuzzyJoin~\cite{autofuzzyjoin}), \underline{\textbf{(J2)} transformation-join} that learns structural transformations before equi-joins can be performed (e.g., Auto-Join~\cite{autojoin}, CST~\cite{cst}, and GXJoin~\cite{omidvartehrani2024gxjoin}), and \underline{\textbf{(J3)} semantic-join} that handles semantic mapping (e.g., Sema-Join~\cite{semajoin} and mapping lookup~\cite{dataxformer}). Large language models (LLMs) go beyond these joins by using their language understanding and world knowledge to enable \underline{\textbf{(J4)} LLM-direct-join},
which supports diverse types of joins by directly judging whether rows match~\cite{comem}. We use an example to explain these methods:

\example{\label{ex:hospital}
Figure~\ref{fig:running} shows sample rows of two hospital tables from a join benchmark~\cite{autofuzzyjoin}, table $S$ with 257 rows and table $T$ with 2{,}424 rows.
The goal is to join the rows that name the same hospital, using the hospital names as join keys.

Let $s_i$ and $t_j$ denote rows of $S$ and $T$, respectively.
For the rows shown, each $(s_i,t_i)$ is a matching row pair that can join.
\underline{\textbf{(J1)} fuzzy-join}: row pairs with $i=1,2,3$ have small textual differences that can be handled by fuzzy-join-like methods.
\underline{\textbf{(J2)} transformation-join}: those with $i=4,5,6$ contain the same hospital names that are structured very differently, which requires transformation-joins that transform join keys from both sides into formats that are ready for equi-joins (fuzzy-joins alone are often unable to handle large string differences in these examples). \underline{\textbf{(J3)} semantic-join}: those with $i=7,8,9$ involve abbreviations or semantically matching values that fuzzy-join or transformation-join often cannot handle.
These joins can be \textit{composed}: the row pairs with $i=10,11$ need a transformation-join (J2) step followed by a fuzzy-join (J1) step, as neither alone is sufficient to join these row pairs.
Finally,
\underline{\textbf{(J4)} LLM-direct-join}: an LLM can directly judge whether rows in any of the four groups refer to the same hospital.
Note that simple equi-join would miss all these row pairs.}

\begin{figure}[t]
\centering
\includegraphics[width=\columnwidth]{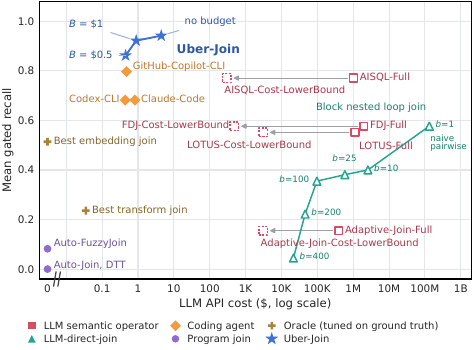}
\caption{
Mean gated recall (i.e., recall if the precision meets the target 0.9, and 0 otherwise) and LLM API cost on \joinhard{}, whose source tables have 250K rows (22 MB) on average.
The baselines run without a budget, and \system{} (ours) with the budgets shown. Hollow markers are estimated costs, and Cost-LowerBound markers remove or minimize the LLM calls on row pairs of the semantic operators~\cite{lotus, fdj, aisql, adaptive-join}. The block nested loop join reads a block of $b$ rows of each table per LLM call, for $b = 1$ to $400$. Claude-Code uses \texttt{claude-sonnet-5}, and other LLM-based methods use \texttt{gpt-5.4}.}
\label{fig:intro_cost}
\end{figure}

As Example~\ref{ex:hospital} shows, J1 to J3 each match the row pairs of one join type, whereas an LLM-direct-join (J4) can match the row pairs of all join types.
LLM-direct-joins also yield high precision (i.e., the fraction of returned row pairs that are correct), which is often desired because incorrect matches propagate errors into downstream analysis~\cite{autofuzzyjoin,semajoin}.
However, an LLM-direct-join pays for every LLM call to judge the row pairs across the two tables, whose number is quadratic in the table size.
Figure~\ref{fig:intro_cost} shows this cost for a modest-sized source table of 250{,}000 rows: even when batching row pairs per LLM call (block nested loop join), an LLM-direct-join costs tens of thousands of dollars or more per dataset.

On the other hand, a fuzzy-, transformation-, or semantic-join (J1 to J3) calls an LLM, if at all, only a few times on sample rows, for example to write a key-rewriting function, regardless of the table sizes. It then joins all rows without LLM calls, and methods such as Auto-FuzzyJoin~\cite{autofuzzyjoin} and Auto-Join~\cite{autojoin} call no LLM at all. However, each of these join types covers only the matches of its own type, and its precision on the given tables can be much lower than that of an LLM-direct-join and is unknown in advance.

These complementary strengths motivate combining inexpensive program-based joins (J1 to J3) with selective LLM-direct-joins (J4) under a shared budget. The goal is to recover as many correct row pairs as possible subject to two constraints: the precision of the join result is at least a desired precision, called the \textit{precision target}, and the total LLM cost is within a \textit{budget}. We refer to this problem as \textit{Precision-Guaranteed Budget-Constrained Join}.

\stitle{Technical challenges.} Precision-Guaranteed Budget-Constrained Join involves two types of uncertainty, and resolving either of them consumes part of the budget.
First, it is unknown which join types suit the given tables. Example~\ref{ex:hospital} shows that several join types and their compositions can each cover a different subset of the matches of one pair of tables, while other pairs of tables can be joined by a fuzzy-join or a transformation-join alone. The effectiveness of a join type on the given tables can only be assessed after it is applied to them. Since applying a join type with an LLM costs budget, a solution cannot try every type extensively and has to explore which types to invest in.

Second, even a suitable join type can produce outputs with widely varying quality. For example, for tables suitable for a transformation-join, a key-rewriting function that an LLM writes from sample rows may extract the wrong field and return incorrect matches. The quality of an output can only be estimated by having an LLM evaluate its pairs. Without such estimates, it is difficult to determine whether an output would satisfy the precision constraint or whether applying the same join type again would be beneficial.

\stitle{Our approach: \system{}.} We present \system{}, a framework that joins two tables with diverse join types within a budget and a precision target. Our key idea is to use LLMs selectively to create join programs and verify their outputs through sampling, while adaptively deciding which joins to pursue under budget and precision constraints. This combines the scalability of program-based joins with the precision of LLMs, while the precision of the final result is guaranteed with high probability.

To do so, \system{} shares the mechanism of a budgeted multi-armed bandit (MAB)~\cite{tranthanh2012budget}, which balances pulls that \textit{explore} arms with unknown rewards against pulls that \textit{exploit} arms known to pay off~\cite{bubeck2012regret}. A planner pulls two levels of arms, one for each challenge: a \textsf{CREATE} arm creates a new join of a join type, and an \textsf{ESTIMATE-EXECUTE} arm executes a created join and estimates the quality of its output with an LLM. Because these arms appear as joins are created and their rewards shrink as other joins cover the same rows, the planner pulls the arm with the most expected new correct pairs per dollar and updates these expectations after each pull. A precision guard accepts an output only if a lower confidence bound on the precision of the join result, computed from the LLM's judgments of sampled pairs, stays at least the precision target. As shown in Figure~\ref{fig:intro_cost}, \system{} achieves a higher gated recall than all baselines, at a far lower cost than the LLM-direct-joins and LLM semantic operators.
Our work makes the following contributions.

\squishlist
\item We formulate \textit{Precision-Guaranteed Budget-Constrained Join}, the problem of creating and executing joins of different join types under uncertainty to maximize recall under budget and precision constraints.
We show it is \textsf{NP}-hard (\S\ref{sec:problem}).
\item We develop \system{}, a greedy framework that uses two levels of arms for creating and executing joins, and establish a statistical guarantee on the precision of the join result (\S\ref{sec:sol}).
\item We introduce \joinhard{}, a benchmark with 250{,}000-row source tables, and show on four benchmarks that \system{} achieves high recall under budget and precision constraints (\S\ref{sec:evaluation}).
\squishend

The implementation of \system{} and our \joinhard{} benchmark have been open-sourced to facilitate future research.\footnote{\url{https://github.com/GuohuiGuan/uber-join}}

\section{Preliminaries}
\label{sec:prelim}

This section defines join types and instances (\S\ref{sec:prelim:types}), their LLM cost (\S\ref{sec:prelim:cost}), and the LLM-oracle that estimates join quality (\S\ref{sec:prelim:oracle}), on which \S\ref{sec:problem} and \S\ref{sec:sol} build.

\subsection{Join Types and Join Program Instances}
\label{sec:prelim:types}

Let $S$ and $T$ be tables with rows $s \in S$ and $t \in T$.
A \emph{join predicate} $\theta(s,t)$ is a Boolean condition that holds if $s$ and $t$ match.
An inner join produces the \emph{join result}
$J = S \bowtie_\theta T = \{(s,t) \in S \times T : \theta(s,t)\}$.

The most common join is the \emph{equi-join}, whose predicate $\theta(s,t) \equiv (s.A = t.B)$ compares the \emph{join keys} $A$ and $B$, which consist of the same number of attributes (one or more).
It misses pairs such as $(s_1,t_1)$ in Figure~\ref{fig:running}, whose keys differ by the word ``The''.

The quality of a join result $J$ can be measured against the set $J^{*} \subseteq S \times T$ of \emph{correct pairs}, such as the pairs of rows that name the same hospital, by its precision and recall~\cite{autofuzzyjoin, autojoin, semajoin}, $\text{precision}(J) = |J \cap J^{*}| / |J|$ and $\text{recall}(J) = |J \cap J^{*}| / |J^{*}|$.

The four join types introduced in \S\ref{sec:intro}
relax equi-join through parameterized join predicates.
We formalize a join type as follows.

\begin{definition}[Join type]\label{def:jointype}
A \emph{join type} $\mathcal{J}$ specifies a
\emph{predicate template} $\theta_{\mathcal{J}}(s,t;\phi)$
and a set $\Phi_{\mathcal{J}}$ of valid parameter
assignments.
\end{definition}

Instantiating a join type, that is, fixing all its parameters, gives a join program instance.

\begin{definition}[Join program instance]\label{def:instance}
A \emph{join program instance}, or \emph{instance},
of a join type $\mathcal{J}$ is a pair
$x = (\mathcal{J}, \phi)$ with a fixed parameter
assignment $\phi \in \Phi_{\mathcal{J}}$.
Its join predicate is
$\theta_x(s,t) = \theta_{\mathcal{J}}(s,t;\phi)$,
and its join result is
$J_x = S \bowtie_{\theta_x} T$.
The set of all instances of $\mathcal{J}$ is
$U(\mathcal{J}) =
\{(\mathcal{J}, \phi) : \phi \in \Phi_{\mathcal{J}}\}$.
\end{definition}

We write $\mathcal{J}_1$ to $\mathcal{J}_4$ for J1 to J4 and $[c]$ for the truth value of a condition $c$.
The four join types are characterized by the following predicate templates, each followed by an example on Figure~\ref{fig:running} that fixes its parameters and thus gives an instance.

\noindent$\bullet$ \textbf{J1: fuzzy-join.}
Its predicate template is
$\theta_{\mathcal{J}_1}(s,t;\phi)
= [d(s.A,t.B) \leq \delta]$,
where $\phi=(d,\delta)$ consists of a distance
function $d$ and a threshold $\delta$.
It tolerates small textual differences.

\example{\label{ex:j1}
Setting $d$ to edit distance and $\delta = 4$ gives $x_a = (\mathcal{J}_1, (d_{\mathrm{edit}}, 4))$, which joins the pairs with $i = 1, 2, 3$.
}

\noindent$\bullet$ \textbf{J2: transformation-join.}
Its predicate template is
$\theta_{\mathcal{J}_2}(s,t;\phi) = [f(s.A) = t.B]$,
where $\phi = f$ is a program that rewrites a key of $S$ into the format of $T$.
It tolerates different layouts.

\example{\label{ex:j2}
Setting $f$ to a program $f_b$ that rewrites ``102 Pok Fu Lam Road, Hong Kong / Queen Mary Hospital'' into ``Queen Mary Hospital (Hong Kong)'' gives the instance $x_b = (\mathcal{J}_2, f_b)$, which joins the pairs with $i = 4, 6$.
The key of $s_5$ is a key-value list instead, so another instance $x_c = (\mathcal{J}_2, f_c)$, whose program $f_c$ extracts the name and the city, joins the pair with $i = 5$.
}

\noindent$\bullet$ \textbf{J3: semantic-join.}
Its predicate template is
$\theta_{\mathcal{J}_3}(s,t;\phi) = [M(s.A) = t.B]$,
where $\phi = M$ is a mapping from values of $S$ to values of $T$.
It tolerates keys related by meaning.

\example{\label{ex:j3}
Setting $M$ to a mapping $M_d$ that sends CHOP to Children's Hospital of Philadelphia, Hospital for Special Surgery to HSS, and Universit\"atsklinikum Heidelberg to Heidelberg University Hospital gives the instance $x_d = (\mathcal{J}_3, M_d)$, which joins the pairs with $i = 7, 8, 9$.
}

\noindent$\bullet$ \textbf{J4: LLM-direct-join.}
Its predicate template is
$\theta_{\mathcal{J}_4}(s,t;\phi) = [t \in T_\phi(s) \text{ and the LLM judges that } s \text{ and } t \text{ match}]$,
where $T_\phi(s)$ is the set of rows of $T$ that the LLM reads together with $s$.
Its parameter assignment $\phi$ consists of the LLM, its prompt, and the rule that selects $T_\phi(s)$ (e.g., all rows of $T$, a block of rows, or the $k$ rows of $T$ nearest to $s$ under an embedding).
Because the LLM reads the rows themselves, J4 is not limited to one kind of difference.

\example{\label{ex:j4}
Setting the rule to the ten rows of $T$ nearest to $s$, with an LLM and a matching prompt, gives the instance $x_e = (\mathcal{J}_4, \phi_e)$, which can join all eleven pairs in Figure~\ref{fig:running}.
}

Join types can also be \textit{composed}, so that the keys of $S$ that one type rewrites are joined by the predicate of the next type.
For example, the composition J2 + J1, a transformation-join followed by a fuzzy-join, has the predicate $d(f(s.A), t.B) \leq \delta$.
With the program $f_b$ of $x_b$ and the distance function and threshold of $x_a$, it gives the instance $x_f$, which joins the pairs with $i = 10, 11$ in Figure~\ref{fig:running}.

Two actions apply to an instance: 1) \emph{creating} it chooses its parameter assignment for the two tables, and 2) \emph{executing} it computes its join result $J_x$; we describe in \S\ref{sec:prelim:cost} how both are done and what they cost.
Finally, for a set of join types $\mathcal{T}$, which can include compositions, $U(\mathcal{T}) = \bigcup_{\mathcal{J} \in \mathcal{T}} U(\mathcal{J})$ denotes the set of their instances, and $J_X = \bigcup_{x \in X} J_x$ is the join result of a set of instances $X$.

\subsection{Cost of Creating and Executing Instances}
\label{sec:prelim:create}
\label{sec:prelim:cost}

\begin{table}[t]
\small
\caption{Number of LLM calls for instance creation and execution, by join type and creation method.}
\label{tab:cost}
\centering
\setlength{\tabcolsep}{4pt}
\begin{tabular}{@{}llcc@{}}
\toprule
\multirow{2}{*}{Type} & \multirow{2}{*}{\makecell[l]{Creation method (J1 to J3) /\\Rows that the LLM reads (J4)}} & \multicolumn{2}{c}{LLM calls} \\
\cmidrule(l){3-4}
 & & Creating & Executing \\
\midrule
\multirow{2}{*}{J1 to J3} & non-LLM-create & 0 & 0 \\
\cmidrule(l){2-4}
 & LLM-create & $O(1)$ & 0 \\
\midrule
\multirow{3}{*}{J4} & one pair per call & 0 & $O(|S| \cdot |T|)$ \\
\cmidrule(l){2-4}
 & blocks of $b$ rows & 0 & $O(|S| \cdot |T| / b^2)$ \\
\cmidrule(l){2-4}
 & $k$ nearest rows & 0 & $O(|S|)$ \\
\bottomrule
\end{tabular}
\end{table}

We consider the monetary cost of LLM calls, which are charged by the input and output tokens of each call; computation without LLM calls, such as executing a program locally, is not counted.
Table~\ref{tab:cost} lists these calls for each join type.

\stitle{Non-LLM-create.} For J1 to J3, an algorithm can search the parameter space for an assignment that fits rows of the two tables, without calling an LLM.
For example, Auto-FuzzyJoin~\cite{autofuzzyjoin} searches distance functions and thresholds for a configuration whose estimated precision meets a target, Auto-Join~\cite{autojoin} and CST~\cite{cst} search programs of string operators, and Sema-Join~\cite{semajoin} and DataXFormer~\cite{dataxformer} take mappings from table corpora and web tables.

\stitle{LLM-create.} Instead of searching without an LLM, an LLM can write the parameter assignment from sample rows: it sets the threshold of a fuzzy-join by judging sample pairs that are close under an embedding, or writes the program of a transformation-join or the mapping of a semantic-join~\cite{nobari2024tabulax}.
This takes a constant number of LLM calls, independent of the table sizes.
In both cases, the created instance is a program whose execution makes no LLM call.
Moreover, the instance depends on the rows that are read, so creating again on the rows of $S$ that existing instances leave unjoined can return a new instance that joins some of these rows.

\stitle{LLM-direct-join.} Unlike creating an instance of J1 to J3, creating an instance of J4 only chooses the LLM, the prompt, and the rule that selects the rows that the LLM reads, which makes no LLM call.
LLM-based entity matching~\cite{matchgpt, comem} and the semantic-operator joins of LLM-powered query processing~\cite{lotus, fdj, aisql, adaptive-join} differ in this rule and in how they reduce the LLM calls.
Executing the instance, however, calls the LLM on the selected rows, and its calls grow with the tables even with blocks of $b$ rows per call, as in a block nested loop join (Table~\ref{tab:cost}).

\subsection{LLM as Oracle to Estimate Join Quality}
\label{sec:prelim:oracle}

\begin{table}[t]
\small
  \caption{Precision of the LLM-oracle (the fraction of the pairs it judges correct that are correct) with two base LLMs.}
  \label{tab:oracle_precision}
  \centering
  \begin{tabular}{lcc}
    \toprule
    Benchmark & \texttt{gpt-5.4} & \texttt{claude-sonnet-5} \\
    \midrule
    \joinhard{} & 0.990 & 0.986 \\
    AFJ & 0.962 & 0.961 \\
    AJ & 0.985 & 0.986 \\
    SJ & 0.986 & 0.977 \\
    \bottomrule
  \end{tabular}
\end{table}

The instances of J1 to J3 need few LLM calls, but their parameters are chosen from sample rows without knowing the correct pairs, so instances of the same type can return join results of very different quality, which shows only in the pairs that they return.

\example{\label{ex:quality}A fuzzy-join (J1) instance with a larger threshold than $x_a$ also joins the pairs with $i = 1, 2, 3$, but it can also join a row of $S$ that names Children's Hospital of Philadelphia with a row of $T$ that names Children's Hospital of Pittsburgh.
A transformation-join (J2) instance whose program returns the street instead of the city rewrites $s_4$ into ``Queen Mary Hospital (102 Pok Fu Lam Road)'', which matches no row of $T$.}

Since the correct pairs $J^{*}$ are unknown in reality, the precision of a join result cannot be computed from the two tables, but it can be estimated by judging individual pairs.
An LLM suits this task, because judging a given pair is a verification task.
The LLM also has the language understanding and world knowledge to recognize the differences of J1 to J3, such as that CHOP abbreviates Children's Hospital of Philadelphia~\cite{petroni2019language}.
On the benchmarks of \S\ref{sec:evaluation}, at least 96\% of the pairs that an LLM judges correct are correct (Table~\ref{tab:oracle_precision}), although its judgments can be wrong.
We therefore use an LLM as a surrogate oracle for the pairs that an instance returns, called the \emph{LLM-oracle}, and the guarantee of \S\ref{sec:sol:precision} assumes that its judgments are correct; the same subsection discusses what happens otherwise.

Using the LLM-oracle, however, also costs LLM calls.
Judging every pair of a join result would cost as much as an LLM-direct-join on these pairs.
Instead, \system{} judges a uniform sample of the pairs, which estimates the precision of the result at a cost that does not grow with the size of the result (\S\ref{sec:sol:precision}).

\section{Problem Statement}
\label{sec:problem}

This section states the problem that \system{} solves and explains the uncertainty that a solution faces (\S\ref{subsec:uncertainty}), and then shows that the problem is \textsf{NP}-hard even without this uncertainty (\S\ref{subsec:hardness}).

\begin{figure*}[t]
\centering
\input{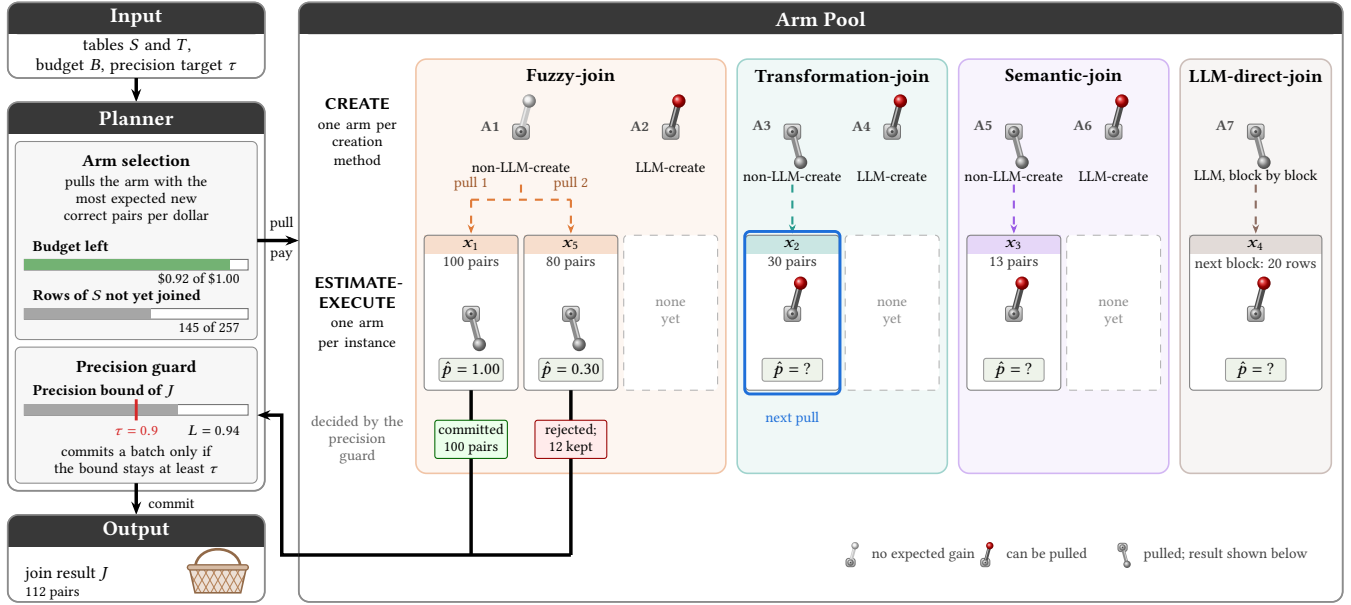}
\caption{Architecture of \system{}, shown before round 8 of Example~\ref{ex:run}. In each round, the planner (left) pulls one arm (right), and the precision guard accepts or rejects the candidate pairs of the pulled instance; if it rejects them, only its sampled pairs judged correct are committed. The blue box marks the next pull.}
\label{fig:arch}
\end{figure*}

\subsection{Precision-Guaranteed Budget-Constrained Join}
\label{subsec:uncertainty}

Since instances of different join types complement each other (\S\ref{sec:prelim:cost}), a join result can combine instances of several types; for example, the instances $x_a$ to $x_d$ and $x_f$ join all eleven pairs in Figure~\ref{fig:running}.

\stitle{Problem statement (Precision-Guaranteed Budget-Constrained Join).}
Given two tables $S$ and $T$, a set of join types $\mathcal{T}$, a precision target $\tau \in (0, 1]$, and a budget $B > 0$, we aim to create and execute a set of instances $X \subseteq U(\mathcal{T})$ that solves
\begin{equation}\label{eq:problem}
\begin{aligned}
\max_{X \subseteq U(\mathcal{T})}\quad & \text{recall}(J_X) \\
\text{s.t.}\quad & \text{precision}(J_X) \geq \tau, \\
& C(X) \leq B,
\end{aligned}
\end{equation}
where $C(X)$ is the cost of creating and executing $X$ (\S\ref{sec:prelim:cost}).
We assume that each row of $S$ has at most one correct match in $T$.

\stitle{Join under uncertainty.} The objective function in~\eqref{eq:problem} cannot be optimized by enumerating $U(\mathcal{T})$, the set of all instances of the given join types, because what each instance contributes is unknown before a solution spends the budget.
First, an instance exists only after it is created from the rows that it reads (\S\ref{sec:prelim:create}), and its join result $J_x$ and its cost are known only after it is created and executed.
Second, even a suitable type can yield instances of different quality (Example~\ref{ex:quality}).
Executing $x$ reveals $J_x$ but not which of its pairs belong to the correct pairs $J^{*}$, and estimating its precision with the LLM-oracle (\S\ref{sec:prelim:oracle}) costs part of the budget $B$.
Because each of these quantities is revealed only by spending the budget, $X$ cannot be chosen in advance; instead, a solution decides step by step which instance to create or execute next, based on what earlier steps revealed.

\subsection{Hardness Even Without Uncertainty}
\label{subsec:hardness}

If all instances had been created and executed and the correct pairs were known, no uncertainty would remain, and the problem would reduce to the following join instance selection problem.

\stitle{Join instance selection.}
Given a finite set of instances $\mathcal{X}$, the join result $J_x$ and the cost $r_x$ of each instance $x \in \mathcal{X}$, the correct pairs $J^{*}$, a precision target $\tau$, and a budget $B$, select $X \subseteq \mathcal{X}$ that maximizes $\text{recall}(J_X)$ subject to $\text{precision}(J_X) \geq \tau$ and $\sum_{x \in X} r_x \leq B$.

\begin{theorem}\label{thm:nphard}
Join instance selection is \textsf{NP}-hard.
\end{theorem}

\begin{proof}[Proof sketch]
Join instance selection contains Budgeted Maximum Coverage (BMC)~\cite{khuller1999budgeted} as a special case. Consider instances whose pairs are all correct: the precision constraint then holds for every nonempty selection, and maximizing recall within the budget selects instances whose join results cover the most correct pairs at a total cost of at most $B$. This is BMC, with correct pairs as elements and join results as subsets. BMC is \textsf{NP}-hard, and so is join instance selection. The proof is in the appendix.
\end{proof}

Since join instance selection is the special case of Precision-Guaranteed Budget-Constrained Join without uncertainty, our problem is \textsf{NP}-hard as well.
For join instance selection, a natural rule greedily selects the instance with the most new correct pairs per dollar, which underlies the classic approximation algorithm for budgeted coverage~\cite{khuller1999budgeted}.

\section{The \system{} Framework}
\label{sec:sol}

\subsection{Solution Overview}
\label{sec:sol:overview}

\begin{algorithm}[t]
\caption{\system{}}
\label{alg:main}
\small
\begin{algorithmic}[1]
\Require tables $S$, $T$; join types $\mathcal{T}$; precision target $\tau$; budget $B$
\Ensure join result $J$
\State $J \gets \emptyset$;\ $\mathcal{A} \gets$ the \textsf{CREATE} arms of the join types in $\mathcal{T}$
\While{an affordable arm in $\mathcal{A}$ is expected to add correct pairs}
  \State $a \gets$ the arm with the most expected new correct pairs per dollar
  \If{$a$ is a \textsf{CREATE} arm}
    \State create instances of its join type
    \State add one \textsf{ESTIMATE-EXECUTE} arm per instance to $\mathcal{A}$
  \Else
    \State estimate the quality of its candidate pairs
    \State commit to $J$ the pairs that the precision guard accepts
  \EndIf
  \State update the remaining budget and the expected gains and costs
\EndWhile
\State \Return $J$
\end{algorithmic}
\end{algorithm}

Constructing a join for tables $S$ and $T$ involves two levels of uncertainty: 1) at the \textit{join-type level}, it is unknown which join types are suitable for the given tables; 2) at the \textit{instance level}, the quality of the created instances is unknown, as the same join type can yield instances with high or low output quality.
Both uncertainties can be reduced only by spending budget, since the quality of an instance, and thus the suitability of its join type, is estimated by having the LLM-oracle judge a sample of its output.
The driving principle of \system{} is therefore to spend the budget where it expects the most new correct pairs per dollar, and to update these expectations with the outcome of each expenditure.
This principle requires a balance between learning about uncertain choices and using what has been learned, which is the problem that a budgeted multi-armed bandit (MAB) addresses.

In the classical stochastic MAB model, each arm has an unknown but fixed reward distribution, and pulling an arm reveals a reward~\cite{bubeck2012regret}.
In a budgeted MAB, each pull also has a cost, and the pulls are limited by a budget~\cite{tranthanh2012budget}.
The observed rewards guide \textit{exploration} of uncertain arms and \textit{exploitation} of promising ones.

\system{} adopts this exploration and exploitation mechanism, but join construction differs from the MAB model in two ways.
First, the reward of creating an instance diminishes rather than following a fixed distribution: each new instance works only on the rows that earlier instances left unjoined, typically the harder ones.
Second, pulling an arm does not reveal its reward: the number of correct pairs that an instance adds is not observed but estimated from a sample of its pairs that the LLM-oracle judges, which costs budget.
These differences lead to two levels of arms: 1) \textsf{CREATE} arms, which create instances of join types, and 2) \textsf{ESTIMATE-EXECUTE} arms, which execute the created instances.
The \textsf{ESTIMATE-EXECUTE} arms are added during the run by pulls of the \textsf{CREATE} arms, as in bandits whose arms are acquired at a cost~\cite{kalvit2021arm}.
Our key novelty is to jointly and adaptively decide which join instances to create and which of their outputs to evaluate under a shared budget.
This combines the scalability of program-based joins with selective LLM verification, while maintaining a statistical precision guarantee on the accumulated result.

Guided by these observations, \system{} adopts the architecture shown in Figure~\ref{fig:arch}, with its main loop summarized in Algorithm~\ref{alg:main}.
In each round, the planner (left) pulls an arm from either of the two levels (right).
We briefly describe the \textsf{CREATE} and \textsf{ESTIMATE-EXECUTE} arms and the planner below.

\stitle{\textsf{CREATE}: exploring join types.}
\system{} explores join types through \textsf{CREATE} arms in the upper row of Figure~\ref{fig:arch}.
Each \textsf{CREATE} arm represents a method for creating instances of a join type (\S\ref{sec:prelim:create}), so a join type can have multiple arms.
Each pull creates instances using rows of $S$ that have not yet been joined.
To predict the benefit of further creation, \system{} uses a diminishing-return heuristic: the best recent gain of an arm serves as an optimistic prediction of its next gain.
Figure~\ref{fig:dimret} supports this heuristic with an experiment that runs each LLM-create fuzzy-, transformation-, and semantic-join method repeatedly on 141 benchmark datasets, each time on the rows that earlier runs left unjoined.
Depending on the method, 85.3\% to 94.6\% of the later runs gain at most as much recall as the better of the two runs before them.

\stitle{\textsf{ESTIMATE-EXECUTE}: exploiting created instances.}
Each instance produced by a \textsf{CREATE} arm introduces an \textsf{ESTIMATE-EXECUTE} arm in the lower row of Figure~\ref{fig:arch}.
For J1 to J3, the instance is executed upon creation without LLM calls, producing candidate pairs without adding them to $J$.
For J4, execution instead occurs block by block when its arm is pulled.
Pulling an \textsf{ESTIMATE-EXECUTE} arm invokes the LLM-oracle to judge a random sample of its remaining candidate pairs, called a batch, to estimate their precision and the number of correct pairs.
The precision guard (\S\ref{sec:sol:precision}) commits the batch only if the lower confidence bound on the precision of $J$, including the batch, remains at least $\tau$.
These estimates also guide further creation from the same \textsf{CREATE} arm.

\stitle{Planner: allocating the budget.}
Since LLM calls for creation, execution, and LLM-oracle checks share a limited budget, the planner greedily pulls the arm with the most expected new correct pairs per dollar across both levels.
This balances exploration and exploitation using predicted gains and costs.
As committed pairs cover rows, the potential gains of other instances decrease, and the planner updates its predictions with new quality estimates.

\begin{table*}[t]
\caption{Trace of \system{} in Example~\ref{ex:run}. $|R|$ and Spent are the values after the round.}
\label{tab:example}
\centering
\footnotesize
\begin{tabularx}{\textwidth}{@{}llXrr@{}}
\toprule
Round & Arm pulled & Outcome & $|R|$ & Spent \\
\midrule
1--4 & \textsf{CREATE} A1, A3, A5, A7 & Untried arms predict all 257 rows; the pulls create $x_1$ to $x_4$ (J1 to J4). & 257 & \$0 \\
5 & \textsf{ESTIMATE-EXECUTE} $x_1$ & 100 candidate pairs; 40 of 40 sampled pairs judged correct; $C^- = 93$, $L = 0.93$; commit 100. & 157 & \$0.04 \\
6 & \textsf{CREATE} A1 & On the remaining rows, A1 creates $x_5$ with a larger threshold. & 157 & \$0.04 \\
7 & \textsf{ESTIMATE-EXECUTE} $x_5$ & 80 candidate pairs; 12 of 40 judged correct; $C^- = 15$, $L = 0.71 < \tau$; reject, but commit the 12 judged correct; A1 drops out. & 145 & \$0.08 \\
8 & \textsf{ESTIMATE-EXECUTE} $x_2$ & 30 candidate pairs, all judged correct; commit 30. & 115 & \$0.12 \\
9 & \textsf{CREATE} A3 & The new instance $x_6$ has no candidate pairs; A3 drops out. & 115 & \$0.12 \\
10--12 & \textsf{CREATE} A2, A4, A6 & The untried LLM-create arms create $x_7$ (J1), $x_8$ (J2), and $x_9$ (J3). & 115 & \$0.72 \\
13 & \textsf{ESTIMATE-EXECUTE} $x_9$ & 55 candidate pairs; 36 of 40 judged correct; $C^- = 44$ of $A = 51$ ($0.86 < \tau$), but $L = 0.93$; commit 51, which also join the rows of 9 candidate pairs of $x_3$. & 64 & \$0.76 \\
14 & \textsf{ESTIMATE-EXECUTE} $x_8$ & 17 candidate pairs, all judged correct; commit 17. & 47 & \$0.80 \\
15 & \textsf{ESTIMATE-EXECUTE} $x_7$ & 12 candidate pairs, all judged; commit the 10 judged correct. & 37 & \$0.84 \\
16 & \textsf{ESTIMATE-EXECUTE} $x_4$ & A block of 20 rows returns 17 pairs, all judged; commit the 16 judged correct; the budget runs out before $x_3$ (4 candidate pairs) is checked. & 21 & \$1.00 \\
\bottomrule
\end{tabularx}
\end{table*}

\example{\label{ex:run}
Table~\ref{tab:example} traces \system{} on the hospital tables of Example~\ref{ex:hospital} with $\tau = 0.9$ and $B = \$1$; the examples in \S\ref{sec:sol:create} to \S\ref{sec:sol:precision} discuss its key rounds.
Of the 257 rows of $S$, 120 differ from their matches by small textual edits (J1), 50 by layout (J2), 55 by abbreviation or translation (J3), and 32 can only be matched by an LLM (J4).
The pool has seven \textsf{CREATE} arms, A1 to A7 in Figure~\ref{fig:arch}: A1, A3, and A5 are the non-LLM-create arms of J1 to J3, A2, A4, and A6 are their LLM-create arms, and A7 is the arm of J4.
For illustration, each pull of a \textsf{CREATE} arm creates a single instance, a pull of an LLM-create arm costs \$0.20, each LLM-oracle check judges 40 sampled pairs and costs \$0.04, and executing the instance of J4 on a block of 20 rows of $S$ costs \$0.12; instances are numbered in their order of creation.
Figure~\ref{fig:arch} shows the state before round 8.
The non-LLM-create fuzzy-join arm A1 has created two instances, $x_1$ and $x_5$: the first was committed, and the second, created on the rows that the first left unjoined, was rejected, so A1 expects no further gain.
The three LLM-create arms, A2, A4, and A6, are untried.
The planner next runs an LLM-oracle check on the J2 instance $x_2$, whose 30 candidate pairs promise more new correct pairs per dollar than any untried arm.
}

\subsection{\texorpdfstring{\textsf{CREATE}: Exploring Join Types}{CREATE: Exploring Join Types}}
\label{sec:sol:create}

A \textsf{CREATE} arm explores a join type by creating instances of it; the quality estimates of these instances (\S\ref{sec:sol:instance}) indicate how effective the type is on the given tables.

\begin{figure}[t]
\centering
\includegraphics[width=\columnwidth]{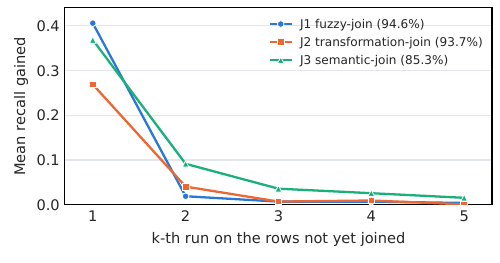}
\caption{Mean recall gained by each run of a join method on the source rows that its earlier runs left unjoined, over the 141 datasets of the four benchmarks. Percentages give the runs after the first whose gain is at most the larger gain of the two runs before them, among the runs where one of these gains is nonzero.}
\label{fig:dimret}
\end{figure}

\stitle{Intuition.}
The gain of a \textsf{CREATE} arm is the number of new correct pairs that its instances can eventually add to $J$.
The planner predicts it with the diminishing-return heuristic (\S\ref{sec:sol:overview}), which guides exploration but does not statistically bound every future instance.

Figure~\ref{fig:dimret} tests this heuristic outside the planner on the 141 datasets of \S\ref{sec:evaluation}.
Each LLM-create method of J1 to J3 (J1 with \texttt{bge-large-en}, no compositions) runs five times, each time on the rows of $S$ whose correct pairs earlier runs missed (by the ground truth), without a precision check or budget.
From the first to the second run, the mean recall gained drops from 0.406 to 0.019 for J1, from 0.269 to 0.041 for J2, and from 0.369 to 0.092 for J3, and later runs add at most 0.036.
Applied to the latest two runs, as in our experiments, the heuristic predicts that a run gains at most as much recall as the better of the two runs before it.
This prediction holds for 94.6\%, 93.7\%, and 85.3\% of the runs after the first for J1, J2, and J3, counting only runs where one of the two earlier runs gained recall.
Semantic-join declines more slowly, because each run has the LLM write a mapping for a new sample of source values.

\stitle{Creating instances and using feedback.}
Pulling a \textsf{CREATE} arm applies its creation method (\S\ref{sec:prelim:create}) to the rows of $S$ not yet joined and to $T$.
For the LLM-direct-join (J4), each LLM and rule that selects the rows the LLM reads (e.g., the 10 rows of $T$ nearest to a row of $S$) forms one \textsf{CREATE} arm, which creates its instance once and without LLM calls (\S\ref{sec:prelim:cost}); the instance is then executed block by block (\S\ref{sec:sol:instance}).
The quality estimates of each instance update the prediction of the \textsf{CREATE} arm that created it.

\example{
In Example~\ref{ex:run}, the first instance $x_1$ of A1, the non-LLM-create fuzzy-join arm, commits 100 pairs in round 5, so A1 predicts 100 new correct pairs for its next pull.
Its second instance $x_5$, created on the rows that $x_1$ left unjoined, is rejected in round 7.
Only the 12 sampled pairs of $x_5$ judged correct are committed, and since the latest pull of A1 has no accepted batch, the prediction of A1 drops to zero.
}

\stitle{Algorithm details.}
\label{sec:sol:residual}
Let $J_\ell$ be the committed result before round $\ell$, and let $R_\ell$ be the rows of $S$ that no pair of $J_\ell$ joins, that is, the rows of $S$ not yet joined.
The planner tracks $|R_\ell|$ in its \emph{Rows of $S$ not yet joined} counter in Figure~\ref{fig:arch}.

For the $h$-th pull of a \textsf{CREATE} arm $a$, let $\mathcal{Q}_{a,h}$ contain the batches of its instances that the precision guard accepted.
For each such batch $Q$, let $|Q|$ be its size and $\hat p_Q$ the fraction of its sampled pairs that the LLM-oracle judged correct (\S\ref{sec:sol:instance}).
Since each row of $S$ has at most one correct match (\S\ref{subsec:uncertainty}), a batch adds at most $|R_\ell|$ new correct pairs.
The record of pull $h$, adjusted to the current rows of $S$ not yet joined, is $g_{a,h}(\ell)=\max_{Q\in\mathcal{Q}_{a,h}}\hat p_Q\min(|Q|,|R_\ell|)$, or zero if $\mathcal{Q}_{a,h}$ is empty.
A pull can create several instances with overlapping outputs (e.g., one for each embedding model of a fuzzy-join), and taking the best batch avoids summing them.

Let $\mathcal{H}_{a,\ell}$ index the latest $K$ pulls of $a$ ($K = 2$ in our experiments and $K = 1$ in Example~\ref{ex:run}).
Its predicted gain is $\widehat G_\ell(a)=|R_\ell|$ if $a$ has not been tried, and $\widehat G_\ell(a)=\max_{h\in\mathcal{H}_{a,\ell}}g_{a,h}(\ell)$ otherwise.
The recent window allows earlier high-gain records to expire.

\subsection{\texorpdfstring{\textsf{ESTIMATE-EXECUTE}: Exploiting Created Instances}{ESTIMATE-EXECUTE: Exploiting Created Instances}}
\label{sec:sol:instance}

Even a join type that works well on the given tables can produce an instance of low output quality (Example~\ref{ex:quality}).
The candidate pairs of an instance do not reveal how many are correct, yet this is needed to guide further creation and enforce the precision target for $J$.

\stitle{Intuition.}
Having the LLM-oracle judge every candidate pair would cost as much as an LLM-direct-join on these pairs (\S\ref{sec:prelim:oracle}).
For example, a \joinhard{} dataset can have 250{,}000 correct pairs, and judging even only these pairs costs about \$108 (embedding top-1, which judges one pair per row of $S$, in Table~\ref{tab:varying_budget}).
An LLM-oracle check therefore judges only a random sample of the batch, from which the precision guard decides whether to commit it (\S\ref{sec:sol:precision}).

\example{
In round 7 of Example~\ref{ex:run}, the instance $x_5$, a fuzzy-join with a larger threshold as in Example~\ref{ex:quality}, has 80 candidate pairs, but only 12 of the 40 sampled pairs are judged correct.
The sample thus reveals a low precision that the number of candidate pairs hides.
The precision guard rejects the batch but still commits the 12 sampled pairs judged correct (\S\ref{sec:sol:precision}).
}

\stitle{Instances of J4.}
Unlike an instance of J1 to J3, an instance of J4 needs LLM calls for every row of $S$ that it processes.
Executing it on all rows not yet joined would thus cost as much as an LLM-direct-join over the whole table (\S\ref{sec:prelim:cost}).
Since J4 can still join pairs that no other type matches, \system{} executes its instance one \emph{block} at a time.
Each pull of its arm executes the instance on the next block of 100 rows of $S$ not yet joined, and the LLM-oracle checks the resulting candidate pairs as for any other instance.
The arm stays in the pool, and the planner pulls it again whenever its next block promises the most new correct pairs per dollar.

\stitle{Algorithm details.}
\textit{Candidate pairs.} When the arm of an instance $x$ of J1 to J3 is introduced, \system{} executes $x$ on the rows of $S$ not yet joined, without LLM calls (\S\ref{sec:prelim:cost}) and without adding pairs to $J$.
This gives the pairs of the join result $J_x$ (Definition~\ref{def:instance}) whose rows of $S$ are not yet joined.
As other batches enter $J$, the remaining candidate pairs of $x$ shrink to $J_{x,\ell}=\{(s,t)\in J_x:s\in R_\ell\}$, whose size bounds the new correct pairs that the arm can add (\S\ref{sec:sol:scheduling}).

\textit{Batch and sample.} When the arm is pulled, its remaining candidate pairs form the batch $Q_k$, the $k$-th batch of the run, which stays fixed while it is checked.
An empty batch is skipped without any LLM-oracle call.
Otherwise, \system{} samples pairs of $Q_k$ uniformly without replacement, and if the LLM-oracle judges $c_k$ of the $m_k$ sampled pairs correct, $\hat p_{Q_k}=c_k/m_k$ estimates its precision.

\textit{Feedback.} The estimate $\hat p_{Q_k}$ updates the predicted quality of the pending instances of the same \textsf{CREATE} arm, and an accepted batch also updates the record of that arm (\S\ref{sec:sol:create}).
Which pairs enter $J$, however, is decided by the precision guard from statistical bounds rather than from $\hat p_{Q_k}$ (\S\ref{sec:sol:precision}).

\subsection{Planner: Allocating the Budget}
\label{sec:sol:scheduling}
\label{sec:sol:costgain}

Because creation, execution, and LLM-oracle checks share the budget, neither extreme works well.
On the one hand, running an LLM-oracle check on every instance as soon as it is created can spend the calls of the LLM-oracle on outputs that add few new pairs.
On the other hand, creating all instances before any LLM-oracle check loses quality feedback and the chance to focus later creation on the rows still unjoined.

\stitle{Intuition.}
The planner compares arms by their predicted new correct pairs per dollar (\S\ref{sec:sol:overview}): the predicted gain of an arm counts the new correct pairs that it can eventually add, and its predicted cost counts the LLM calls still needed to obtain these pairs and have the LLM-oracle judge them.

\example{
In round 13 of Example~\ref{ex:run}, four instances are pending: $x_3$ of A5 with 13 candidate pairs, $x_7$ of A2 with 12, $x_8$ of A4 with 17, and $x_9$ of A6, the LLM-create semantic-join arm, with 55.
The planner runs an LLM-oracle check on $x_9$ first, since it promises the most new correct pairs per dollar.
Committing $x_9$ also joins the rows of 9 correct candidate pairs of $x_3$, which leaves $x_3$ with only 4 candidate pairs.
Its predicted gain per dollar then falls below that of a block of the LLM-direct-join instance $x_4$ of A7, which the planner pulls with the last \$0.16 in round 16.
}

\stitle{Gain and cost predictions.}
A \textsf{CREATE} arm uses its diminishing-return prediction (\S\ref{sec:sol:create}).
For the \textsf{ESTIMATE-EXECUTE} arm of an instance $x$ of J1 to J3, whose candidate pairs are known when the arm is introduced (\S\ref{sec:sol:instance}), the predicted gain is $\widehat G_\ell(x)=|J_{x,\ell}|\,\widetilde p_x$, where $x$ also denotes its arm.
Here, $\widetilde p_x$ is the observed precision of the instance of the same \textsf{CREATE} arm whose sample the LLM-oracle judged most recently, or 1 if there is none.
For an instance of J4, whose next block has not been executed, $|J_{x,\ell}|$ is replaced by the size of the next block times the fraction of the rows of its earlier blocks for which its LLM returned a match.
Before its first block, it is replaced by the size of the next block.
Similarly, $\widetilde p_x$ is the precision observed on its earlier blocks, or 1 before its first block.
The prediction optimistically assumes the batch is accepted.

The cost $\widehat r_\ell(a)$ estimates the LLM charges still needed for the predicted contribution.
For a \textsf{CREATE} arm, it includes creation and an anticipated LLM-oracle check of an output.
For an \textsf{ESTIMATE-EXECUTE} arm, it includes the LLM-oracle check and, for J4, the LLM calls of executing the next block.
Including an LLM-oracle check in both arm types' predicted costs makes their gain-per-dollar ratios well defined.
Paid costs are not charged again.

For an operation that calls an LLM, \system{} estimates the cost of its first run from its prompt tokens and an expected output length, and the cost of a later run from the previous run; an LLM-oracle check is priced at the average cost of the samples judged so far.
This price assumes that one sample suffices, and further samples (\S\ref{sec:sol:precision}) are charged when drawn.

\stitle{Algorithm details.}
Let $B_\ell$ be the remaining budget.
Among the available arms $\mathcal{A}_\ell$ with positive predicted gain and predicted cost at most $B_\ell$, the planner greedily chooses $a_\ell \in \operatorname*{arg\,max}_{a \in \mathcal{A}_\ell} \widehat G_\ell(a) / \widehat r_\ell(a)$.
An optimistic prediction keeps untried \textsf{CREATE} arms eligible for exploration~\cite{tranthanh2012budget}.
A \textsf{CREATE} arm waits until at least one instance from its previous pull has had an LLM-oracle check, unless those instances have no remaining candidate pairs.
This lets feedback from one pull inform further creation.

The planner subtracts actual charges from $B_\ell$ and stops when no available arm has positive predicted gain within the remaining budget.
Before each LLM call, it caps the output tokens at the amount the budget can cover, and it does not send the call if the budget cannot cover the input.

\subsection{Statistical Precision Guarantee}
\label{sec:sol:precision}
\label{sec:sol:estimate}

This subsection describes the precision guard, which decides whether a batch can enter $J$.
It then shows that, if the verdicts of the LLM-oracle are correct, the final $J$ satisfies $\Pr[\text{precision}(J) \geq \tau] \geq 1-\alpha$ for a confidence level $1-\alpha$, which the quality estimates of the planner alone do not ensure.

\stitle{Checking each batch.}
From the sample that an LLM-oracle check judges (\S\ref{sec:sol:instance}), \system{} computes a \emph{lower confidence bound} on the precision of the batch, which exceeds the true precision only with a small error probability.
Checking $J$ only once, after the last round, would put the whole run at risk: if this check failed, $J$ could not be returned with the guarantee.
\system{} therefore checks each new batch before committing it, so that $J$ keeps the guarantee throughout the run and later rounds can react to rejected batches.
Each check can be wrong with its error probability, so the guarantee for the final $J$ must cover the errors of all checks.

\stitle{Naive approaches.}
\textit{Checking each batch alone}, that is, committing it only if the lower confidence bound on its own precision is at least $\tau$, does not let batches compensate for each other.
With $\tau = 0.9$, two batches of equal size with precision 0.95 and 0.85 together have precision 0.9, yet the second is rejected with high probability.
\textit{Checking $J$ as a whole}, with a fresh sample from $J$ and the new batch at every check, lets batches compensate but discards the verdicts of earlier samples.
Moreover, a union bound must split $\alpha$ over all checks, for example $\alpha/2^j$ for the $j$-th check.
With $\alpha = 0.1$, a sample of 100 pairs, all judged correct, from at least 10{,}000 pairs then gives a bound of about 0.97 at the first check but below 0.9 from the twelfth check on.
From then on, no batch can be committed at $\tau = 0.9$ unless the samples grow.

\stitle{Incremental precision guard.}
We propose an \emph{incremental precision guard}, which bounds the correct pairs of each batch once and reuses this bound at every later check.
From the sample of a batch, it computes a lower confidence bound on the number of correct pairs that the batch would add to $J$, and it records the bound when the batch is committed.
The batches in $J$ share no row of $S$ (\S\ref{sec:sol:instance}).
Hence, whenever all recorded bounds hold, their sum divided by the number of pairs in $J$ is a lower bound on the precision of $J$.
A new batch is committed only if this bound, with the batch included, is at least $\tau$, so batches compensate for each other without new verdicts.
In round 13 of Example~\ref{ex:run}, the bound of the batch of $x_9$ alone is 0.86, below $\tau = 0.9$, but the batch is committed because the bound on $J$ with it included is 0.93 (Table~\ref{tab:example}).

\stitle{Algorithm details.}
\textit{Bound of a batch.} Let the batch $Q_k$ (\S\ref{sec:sol:instance}) have $N_k = |Q_k|$ pairs, of which an unknown number $C_k$ are correct.
\system{} draws $m_k = \min(n, N_k)$ of its pairs uniformly without replacement, with $n = 100$, and the LLM-oracle judges $c_k$ of them correct.
If the batch is \emph{fully judged}, that is, $m_k = N_k$, then $C_k = c_k$ is known, and its bound is $C^{-}_k = c_k$.
Otherwise, it is \emph{partially sampled}, and its bound at confidence level $1-\alpha_k$ is $C^{-}_k = \min\{C' \in \{0, \dots, N_k\} : \Pr[Y \geq c_k] > \alpha_k\}$, where $Y \sim \mathrm{Hypergeometric}(N_k, C', m_k)$.
If the batch had fewer than $C^{-}_k$ correct pairs, at least $c_k$ sampled pairs would be judged correct with probability at most $\alpha_k$.
If the output-token cap (\S\ref{sec:sol:scheduling}) cuts an LLM-oracle response short, its verdicts are discarded, and the batch is decided by its earlier checks, or rejected if it has none.

\textit{Commit rule.} Let $A_k \leq N_k$ be the number of pairs that batch $k$ would add to $J$; a batch with $A_k = 0$ is skipped.
For each committed batch, the precision guard records $(C^{-}_k, A_k)$.
It commits a new batch $k$ if $L = (\sum_i C^{-}_i + C^{-}_k) / (\sum_i A_i + A_k) \geq \tau$, where $i$ ranges over the batches in $J$, and an empty $J$ has precision 1.
If $L < \tau$ and the batch is partially sampled, the LLM-oracle judges $\min(n, N_k - m_k)$ more pairs, drawn uniformly without replacement from those not yet judged, and the guard recomputes $C^{-}_k$, $A_k$, and $L$ over all judged pairs.
This repeats at most twice and only if the budget allows, and the batch is rejected if $L < \tau$ after its last sample.

\textit{Error probabilities.} A fully judged batch needs no error probability, because its bound is exact.
If batch $k$ is the $j$-th partially sampled batch, whether committed or rejected, each of its at most three bounds uses $\alpha_k = \alpha / (3 \cdot 2^j)$.
The error probabilities of all possible bounds then sum to at most $\sum_{j \geq 1} \alpha / 2^j = \alpha$.

\textit{Reusing the verdicts.} A committed batch leaves out its sampled pairs judged incorrect and adds $A_k = N_k - (m_k - c_k)$ pairs.
Since the left-out pairs are incorrect, $C^{-}_k$ still bounds its correct pairs, and a fully judged batch adds only its $c_k$ pairs judged correct.
A rejected batch still commits its sampled pairs judged correct, recorded with $C^{-}$ and $A$ both equal to their number; this record has ratio 1 and so cannot lower $L$ below $\tau$.

\begin{proposition}[Precision guarantee]\label{prop:gate}
Assume that the verdicts of the LLM-oracle are correct. Then the final join result $J$ of \system{} satisfies $\Pr[\text{precision}(J) \geq \tau] \geq 1-\alpha$.
\end{proposition}

{
    \begin{proof}[Proof sketch]
    Fix the history of the run before the first sample of a partially sampled batch. The pairs of the batch are then fixed, and its samples read prefixes of fixed lengths of one uniformly random order of these pairs, so each possible bound exceeds $C_k$ with probability at most its error probability, whether or not the run draws the sample. By the union bound over all possible bounds of all partially sampled batches, whose error probabilities sum to at most $\alpha$, with probability at least $1-\alpha$ none of these bounds exceeds the $C_k$ of its batch. On this event, each record in $J$ has at least $C^{-}_k$ correct pairs among its $A_k$ pairs, because the left-out pairs are incorrect and the pairs judged correct are correct. The records do not overlap, so the precision of $J$ is at least $\sum_k C^{-}_k / \sum_k A_k$, which every commit keeps at least $\tau$. The full proof is in the appendix.
    \end{proof}
}

\example{
At the end of Example~\ref{ex:run}, with $\alpha = 0.1$, the final $J$ has 236 pairs, and the recorded bounds give $L = 222/236 = 0.94 \geq \tau$ (Table~\ref{tab:example}).
Only the batches of $x_1$, $x_5$, and $x_9$ were partially sampled, with $\alpha_k = \alpha/(3 \cdot 2^j) \approx 0.017$, $0.008$, and $0.004$ for $j = 1, 2, 3$.
Hence, if the verdicts of the LLM-oracle are correct, all three bounds hold with probability at least $1 - \alpha = 0.9$, and then $\text{precision}(J) \geq 0.94$, as in the proof of Proposition~\ref{prop:gate}.
In this example, 232 of the 236 pairs are correct, a precision of 0.98 and a recall of 0.90.
}

\stitle{Discussion.} The guarantee of Proposition~\ref{prop:gate} does not depend on which batches \system{} produces or in which order it checks them, so the planner (\S\ref{sec:sol:scheduling}) can rely on estimates without affecting it.
The assumption that the verdicts of the LLM-oracle are correct is a practical limitation: the LLM-oracle can misjudge a pair, and on the benchmarks, 96\% to 99\% of the pairs that it judges correct are correct (Table~\ref{tab:oracle_precision}).
When it misjudges pairs, the guarantee holds for the precision as judged by the LLM-oracle, and the true precision of $J$ can fall below $\tau$.
Experiment~5 observes this only for targets $\tau \geq 0.95$, near the precision of the LLM-oracle (0.962 on AFJ).

\section{Evaluation}
\label{sec:evaluation}

\subsection{Evaluation Setup}
\label{sec:exp-setup}

\stitle{Benchmarks.}
\system{} is evaluated on three benchmarks from prior join work and on \joinhard{}, which scales datasets of two of them to 250{,}000 rows.
\circled{1} AFJ~\cite{autofuzzyjoin}: 50 datasets that join the entity names of a 2013 DBpedia snapshot with those of later snapshots, where the names of one entity differ by small edits, added disambiguators, or renaming; exact matches are removed.
\circled{2} AJ~\cite{autojoin}: 31 pairs of web tables on 17 topics whose join columns hold the same entities in different formats (e.g., \texttt{John Smith} vs.\ \texttt{Smith, J.}).
\circled{3} SJ~\cite{semajoin}: 50 pairs of columns from web tables and lists whose values are related by meaning (e.g., a US state and its capital).
The source and target tables have 358 and 3{,}295 rows on average in AFJ, 116 and 179 in AJ, and 129 and 66 in SJ.
\circled{4} \joinhard{}: ten datasets whose source tables have 250{,}000 rows.
Five are adapted from AFJ: each keeps the rows and ground-truth pairs of its AFJ dataset, adds as many new pairs whose source value is a fixed transformation of the target value (word-order rearrangement or entity re-serialization), and fills the source table with entity names from unrelated datasets, while the target table keeps its original size (3{,}003 to 5{,}774 rows).
The other five are adapted from AJ and join two 250{,}000-row tables one to one on composite keys: each half of the source keys is rewritten by one of two composed transformation rules, and half of the source rows also carry a transposition of two adjacent letters.
The rules are listed in the open-source repository.

\stitle{Baselines.}
\system{} is compared with 16 baselines in five families.
\emph{(i) Program joins}, which make no LLM call:
\circled{1} Auto-FuzzyJoin~\cite{autofuzzyjoin} tunes the distance function and the threshold of a fuzzy-join to a target precision without labeled pairs;
\circled{2} Auto-Join~\cite{autojoin} and \circled{3} CST~\cite{cst} learn string-transformation programs, CST with a faster search;
\circled{4} DTT~\cite{dtt} transforms each value with a fine-tuned language model.
\emph{(ii) LLM semantic operators}, each given a join condition that the base LLM writes from sample rows:
\circled{5} LOTUS~\cite{lotus} calls the LLM only on pairs whose embedding similarity falls between two thresholds learned from LLM labels;
\circled{6} FDJ~\cite{fdj} prunes pairs with a rule over features that the LLM extracts from each row;
\circled{7} AISQL~\cite{aisql} classifies each row of $S$ against the join column of $T$ with one LLM call per row;
\circled{8} Adaptive-Join~\cite{adaptive-join} prompts the LLM with blocks of rows of both tables sized to the context window.
\emph{(iii) Coding agents}: \circled{9} Codex-CLI~\cite{codexcli}, \circled{10} Claude-Code~\cite{claudecode}, and \circled{11} GitHub-Copilot-CLI~\cite{copilotcli} write and run a join script in an isolated sandbox, given the two tables and $\tau$.
\emph{(iv) Oracles}, which use the ground truth to pick, among their candidate programs, the one with the highest recall at a precision of at least $\tau$:
\circled{12} \emph{best-embedding-join} over three embedding models (\texttt{bge-large-en}~\cite{xiao2024cpack}, \texttt{mxbai-large}~\cite{mxbai2024}, \texttt{arctic-large}~\cite{merrick2024arctic}) and their thresholds, and
\circled{13} \emph{best-transform-join} over the programs that the base LLM writes and those that Auto-Join learns.
\emph{(v) LLM-direct-joins}: the base LLM picks, for each row of $S$, its match or none among the \circled{14} one or \circled{15} ten rows of $T$ nearest under \texttt{bge-large-en} (\emph{embedding top-1} and \emph{top-10}), or \circled{16} reads pairs of 100-row blocks of $S$ and $T$ (\emph{block nested loop join}).
The semantic operators, the coding agents, and the LLM-direct-joins use the base LLM of \system{} and stop when $B$ is spent; the program joins and the oracles run without a budget.

\begin{table*}[t]
\footnotesize
  \caption{(Exp.~1) Gated recall $GR_\tau$ of \system{} and the baselines on four benchmarks under a \$1 budget with two base LLMs. N/A marks a coding agent that does not run on that base LLM, and timeout a method that finishes no dataset within 24 hours. In this and the following tables, the \protect\rkfirst{first}, \protect\rksecond{second}, and \protect\rkthird{third} best results per column are shaded, and the oracles are not ranked.
  }
  \label{tab:e2e-gr-nocc}
  \centering
  \setlength{\tabcolsep}{4pt}
\begin{tabular}{ll cccc c cccc}
    \toprule
     & & \multicolumn{4}{c}{\texttt{gpt-5.4}} & & \multicolumn{4}{c}{\texttt{claude-sonnet-5}} \\
    \cmidrule(lr){3-6} \cmidrule(lr){8-11}
    \multicolumn{2}{l}{Method} & \joinhard{} & AFJ & AJ & SJ & & \joinhard{} & AFJ & AJ & SJ \\
    \midrule
    \multirow{4}{*}{Program}
     & Auto-FuzzyJoin & 0.082 & 0.347 & 0.647 & 0.015 & & 0.082 & 0.347 & 0.647 & 0.015 \\
     & Auto-Join & 0.000 & 0.036 & 0.769 & 0.012 & & 0.000 & 0.036 & 0.769 & 0.012 \\
     & CST & timeout & 0.021 & 0.446 & 0.006 & & timeout & 0.021 & 0.446 & 0.006 \\
     & DTT & 0.000 & 0.000 & 0.466 & 0.000 & & 0.000 & 0.000 & 0.466 & 0.000 \\
    \midrule
    \multirow{4}{*}{\makecell[l]{LLM semantic\\operator}}
     & LOTUS & 0.010 & 0.034 & 0.678 & 0.261 &  & 0.000 & 0.046 & 0.555 & 0.243 \\
     & FDJ & 0.000 & 0.003 & 0.461 & 0.343 &  & 0.000 & 0.008 & 0.475 & 0.339 \\
     & AISQL & 0.000 & 0.249 & 0.840 & \rkthird{0.935} &  & 0.000 & 0.192 & 0.903 & 0.929 \\
     & Adaptive-Join & 0.000 & 0.048 & 0.616 & 0.437 &  & 0.000 & 0.079 & 0.267 & 0.119 \\
    \midrule
    \multirow{3}{*}{Coding agent}
     & Codex-CLI & \rksecond{0.771} & \rkthird{0.681} & \rksecond{0.992} & \rkfirst{0.956} &  & N/A & N/A & N/A & N/A \\
     & Claude-Code & N/A & N/A & N/A & N/A &  & \rksecond{0.695} & 0.568 & \rkthird{0.992} & \rkfirst{0.959} \\
     & GitHub-Copilot-CLI & \rkthird{0.698} & 0.604 & 0.923 & 0.905 &  & \rkthird{0.527} & 0.488 & \rksecond{0.995} & \rkthird{0.941} \\
    \midrule
    \multirow{2}{*}{Oracle}
     & best-embedding-join & 0.514 & 0.672 & 0.916 & 0.379 &  & 0.514 & 0.672 & 0.916 & 0.379 \\
     & best-transform-join & 0.236 & 0.173 & 0.831 & 0.052 &  & 0.270 & 0.157 & 0.869 & 0.059 \\
    \midrule
    \multirow{3}{*}{LLM-direct-join}
     & Embedding top-1 & 0.013 & 0.580 & 0.825 & 0.409 &  & 0.012 & \rkthird{0.654} & 0.887 & 0.408 \\
     & Embedding top-10 & 0.003 & \rksecond{0.707} & 0.936 & 0.686 &  & 0.003 & \rksecond{0.692} & 0.973 & 0.678 \\
     & Block nested loop join & 0.000 & 0.546 & \rkthird{0.939} & 0.900 &  & 0.000 & 0.384 & 0.893 & 0.745 \\
    \specialrule{\lightrulewidth}{\aboverulesep}{0pt}
    \rowcolor{gray!15}
     & \textbf{\system{}}\sysrowstrut & \rkfirst{\textbf{0.922}} & \rkfirst{\textbf{0.785}} & \rkfirst{\textbf{0.996}} & \rksecond{\textbf{0.937}} & & \rkfirst{\textbf{0.949}} & \rkfirst{\textbf{0.753}} & \rkfirst{\textbf{0.996}} & \rksecond{\textbf{0.951}} \\
    \specialrule{\heavyrulewidth}{0pt}{\belowbottomsep}
  \end{tabular}
\end{table*}

\begin{table*}[t]
\footnotesize
  \caption{(Exp.~2) Gated recall $GR_\tau$ of \system{} and the LLM-based baselines on \joinhard{} and AFJ with six base LLMs under a \$1 budget. N/A marks a coding agent that does not run on that base LLM.}
  \label{tab:varying_llms}
  \centering
  \setlength{\tabcolsep}{3pt}
\resizebox{\textwidth}{!}{%
\begin{tabular}{llcccccccccccc}
    \toprule
     & & \multicolumn{6}{c}{\joinhard{}} & \multicolumn{6}{c}{AFJ} \\
    \cmidrule(lr){3-8} \cmidrule(lr){9-14}
    \multicolumn{2}{l}{Method} & \texttt{gpt-5.4} & \makecell{\texttt{claude-}\\\texttt{sonnet-5}} & \makecell{\texttt{gpt-5.4-}\\\texttt{mini}} & \makecell{\texttt{DeepSeek-}\\\texttt{V4-Pro}} & \makecell{\texttt{gpt-oss-}\\\texttt{120b}} & \makecell{\texttt{Mistral-}\\\texttt{Large-3}} & \texttt{gpt-5.4} & \makecell{\texttt{claude-}\\\texttt{sonnet-5}} & \makecell{\texttt{gpt-5.4-}\\\texttt{mini}} & \makecell{\texttt{DeepSeek-}\\\texttt{V4-Pro}} & \makecell{\texttt{gpt-oss-}\\\texttt{120b}} & \makecell{\texttt{Mistral-}\\\texttt{Large-3}} \\
    \midrule
    \multirow{4}{*}{\makecell[l]{LLM semantic\\operator}}
         & LOTUS & 0.010 & 0.000 & 0.000 & 0.000 & 0.000 & 0.000 & 0.034 & 0.046 & 0.132 & 0.058 & 0.027 & 0.000 \\
         & FDJ & 0.000 & 0.000 & 0.000 & 0.000 & 0.000 & 0.000 & 0.003 & 0.008 & 0.018 & 0.009 & 0.005 & 0.024 \\
           & AISQL & 0.000 & 0.000 & 0.000 & 0.000 & 0.000 & 0.000 & 0.249 & 0.192 & 0.213 & 0.349 & 0.227 & 0.131 \\
         & Adaptive-Join & 0.000 & 0.000 & 0.000 & 0.000 & 0.000 & 0.000 & 0.048 & 0.079 & 0.000 & 0.020 & 0.271 & 0.016 \\
    \midrule
    \multirow{3}{*}{Coding agent}
        & Codex-CLI & \rksecond{0.771} & N/A & \rksecond{0.629} & N/A & N/A & N/A & \rkthird{0.681} & N/A & 0.543 & N/A & N/A & N/A \\
        & Claude-Code & N/A & \rksecond{0.695} & N/A & N/A & N/A & N/A & N/A & 0.568 & N/A & N/A & N/A & N/A \\
        & GitHub-Copilot-CLI & \rkthird{0.698} & \rkthird{0.527} & \rkthird{0.280} & N/A & N/A & N/A & 0.604 & 0.488 & 0.238 & N/A & N/A & N/A \\
    \midrule
    \multirow{3}{*}{LLM-direct-join}
             & Embedding top-1 & 0.013 & 0.012 & 0.002 & \rksecond{0.013} & \rksecond{0.079} & \rksecond{0.005} & 0.580 & \rkthird{0.654} & \rksecond{0.611} & \rkthird{0.651} & \rkthird{0.684} & \rkthird{0.594} \\
            & Embedding top-10 & 0.003 & 0.003 & 0.001 & \rkthird{0.003} & \rkthird{0.021} & \rkthird{0.002} & \rksecond{0.707} & \rksecond{0.692} & \rkthird{0.609} & \rksecond{0.723} & \rksecond{0.712} & \rksecond{0.664} \\
          & Block nested loop join & 0.000 & 0.000 & 0.000 & 0.000 & 0.001 & 0.000 & 0.546 & 0.384 & 0.025 & 0.389 & 0.625 & 0.060 \\
    \specialrule{\lightrulewidth}{\aboverulesep}{0pt}
    \rowcolor{gray!15}
        & \textbf{\system{}}\sysrowstrut & \rkfirst{\textbf{0.922}} & \rkfirst{\textbf{0.949}} & \rkfirst{\textbf{0.850}} & \rkfirst{\textbf{0.724}} & \rkfirst{\textbf{0.874}} & \rkfirst{\textbf{0.670}} & \rkfirst{\textbf{0.785}} & \rkfirst{\textbf{0.753}} & \rkfirst{\textbf{0.743}} & \rkfirst{\textbf{0.795}} & \rkfirst{\textbf{0.784}} & \rkfirst{\textbf{0.699}} \\
    \specialrule{\heavyrulewidth}{0pt}{\belowbottomsep}
  \end{tabular}}
\end{table*}

\stitle{Evaluation metrics.}
The headline metric is the \emph{gated recall} $GR_{\tau} = R \cdot \mathbf{1}[P \geq \tau]$, the recall $R$ of a result counted only when its precision $P$ reaches the target $\tau$.
$P$ and $R$ are computed against the ground truth, and a dataset on which a method commits no pair counts as $P = R = F_1 = 0$.
Some experiments also report $P$, $R$, $F_1$, the API cost, billed at the providers' per-token prices with cached input tokens at their reduced prices, and the runtime.
Each value is the mean over the datasets of a benchmark unless stated otherwise.
A baseline that exceeds 24 hours on a dataset is averaged over the datasets that it finishes, and \emph{timeout} marks one that finishes none.

\stitle{Implementation details.}
\system{} is implemented in Python 3.11 and runs on a server with two Intel Xeon 6767P CPUs (128 cores, 256 threads), 1~TB of memory, and two NVIDIA H200 GPUs, which compute the embeddings.
The baselines use the authors' code where it is public (the \texttt{autofj} package of Auto-FuzzyJoin~\cite{autofuzzyjoin}, the binary release of Auto-Join~\cite{autojoin}, and the repositories of CST~\cite{cst}, DTT~\cite{dtt}, LOTUS~\cite{lotus}, FDJ~\cite{fdj}, and Adaptive-Join~\cite{adaptive-join}, with a budget stop added to the semantic operators), and AISQL~\cite{aisql} is re-implemented from its paper.
Codex-CLI runs with \texttt{gpt-5.4} and \texttt{gpt-5.4-mini}, Claude-Code with \texttt{claude-sonnet-5}, and GitHub-Copilot-CLI with all three, none with access to the ground truth.
The prompts and configurations are in the open-source repository.

\stitle{Default join types in \system{}.}
The default pool has 17 \textsf{CREATE} arms (\S\ref{sec:sol:create}).
\circled{1} \emph{J1: fuzzy-join} (two arms): an LLM-create arm whose pull creates one instance per embedding model (\texttt{bge-large-en}, \texttt{mxbai-large}, \texttt{arctic-large}) with a threshold set from 100 pairs that the base LLM judges, and a non-LLM-create arm, Auto-FuzzyJoin with a precision target of 0.95.
\circled{2} \emph{J2: transformation-join} (two arms): an LLM-create arm whose instances are programs that the base LLM writes from sample rows, and a non-LLM-create arm, Auto-Join.
\circled{3} \emph{J3: semantic-join} (one arm): an LLM-create arm whose instances are value mappings that the base LLM writes from sample rows.
\circled{4} \emph{Composition}: the LLM-create arms of J2 and J3 also follow each instance with a J1 fuzzy-join over \texttt{bge-large-en}, creating J2 + J1 and J3 + J1 instances.
\circled{5} \emph{J4: LLM-direct-join} (12 arms): the three rules of the LLM-direct-join baselines, each with \texttt{gpt-5.4}, \texttt{gpt-5.4-mini}, \texttt{claude-sonnet-5}, and \texttt{claude-haiku-4-5}; each instance executes on blocks of 100 rows of $S$ not yet joined (\S\ref{sec:sol:instance}).
Auto-FuzzyJoin and Auto-Join stop after 300 seconds and are left out on \joinhard{}, where they take hours.

\stitle{Parameter settings.}
The precision target is $\tau = 0.90$ at confidence level $1-\alpha = 0.90$, and the budget is $B = \$1$ per dataset.
The base LLM writes the programs and the mappings, sets the fuzzy thresholds, and serves as the LLM-oracle, which judges samples of at most $n = 100$ pairs (\S\ref{sec:sol:precision}).
The base LLM is \texttt{gpt-5.4} unless stated otherwise, and \textsf{CREATE} arms predict their gains from their latest $K = 2$ pulls (\S\ref{sec:sol:create}).
All datasets use the same configuration.

\subsection{Experimental Results}
\label{sec:exp-results}

\newcounter{exp}
\setcounter{exp}{1}

\stitle{Experiment \theexp: End-to-end comparison.}
Table~\ref{tab:e2e-gr-nocc} compares \system{} and the baselines on the four benchmarks under a \$1 budget, with \texttt{gpt-5.4} and \texttt{claude-sonnet-5} as the base LLM.

\textbf{Observation \theexp. \system{} achieves the highest gated recall on three of the four benchmarks, with the largest lead on \joinhard{}, where the semantic operators and the LLM-direct-joins recover almost no pairs within \$1.}
\system{} is first on \joinhard{}, AFJ, and AJ with both base LLMs and second on SJ, where it trails the coding agents by at most 0.019. Its lead is largest on \joinhard{} (0.922 and 0.949, against at most 0.771), where the semantic operators and the LLM-direct-joins reach at most 0.013 within \$1 and the best single program selected with the ground truth reaches 0.514. Several baselines exceed 0.9 on SJ, but none reaches 0.78 on \joinhard{} or AFJ.

\stepcounter{exp}
\stitle{Experiment \theexp: Varying base LLMs.}
Table~\ref{tab:varying_llms} reports the gated recall of \system{} and the LLM-based baselines on \joinhard{} and AFJ under a \$1 budget with six base LLMs: \texttt{gpt-5.4}, \texttt{claude-sonnet-5}, \texttt{gpt-5.4-mini}, and the open-weight \texttt{DeepSeek-V4-Pro}, \texttt{gpt-oss-120b}, and \texttt{Mistral-Large-3}.

\textbf{Observation \theexp. The advantage of \system{} holds across base LLMs, including smaller and open-weight ones.}
\system{} has the highest gated recall with every base LLM on both benchmarks: 0.670 to 0.949 on \joinhard{} against at most 0.771, and 0.699 to 0.795 on AFJ against at most 0.723. With the open-weight LLMs, where no coding agent runs, every baseline stays below 0.08 on \joinhard{}.

\stepcounter{exp}
\stitle{Experiment \theexp: Cost and runtime.}
Figure~\ref{fig:cost_runtime} reports the API cost and runtime of \system{} per dataset, split by \textsf{CREATE} and \textsf{ESTIMATE-EXECUTE}, on all benchmarks with both base LLMs.

\textbf{Observation \theexp. Most of the cost goes to creating instances, and the checks behind the precision guarantee add little.}
\system{} spends close to \$1 per dataset on \joinhard{}, about \$0.70 on AFJ, and less than \$0.30 on AJ and SJ, where it stops once no arm is expected to add correct pairs. \textsf{CREATE} takes 45\% to 75\% of the cost and 51\% to 90\% of the runtime. \textsf{ESTIMATE-EXECUTE}, which includes the checks behind the precision guarantee, costs at most \$0.33 per dataset, because each sample has at most 100 pairs.

\begin{figure}[t]
\centering
\includegraphics[width=\columnwidth]{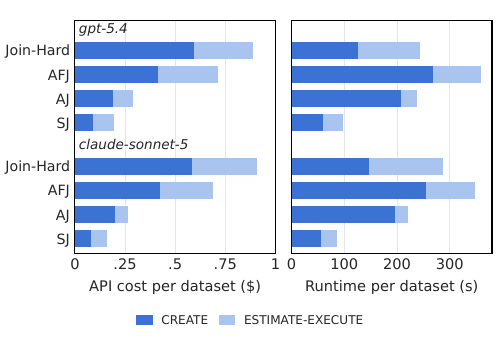}
\caption{(Exp.~3) API cost and runtime of \system{} per dataset by action on four benchmarks with two base LLMs.
}
\label{fig:cost_runtime}
\end{figure}

\stepcounter{exp}
\stitle{Experiment \theexp: Varying budget.}
Table~\ref{tab:varying_budget} reports the gated recall of \system{} and the LLM-based baselines on \joinhard{} under budgets from \$0.5 to \$16 and with an unlimited budget.
The semantic operators and the LLM-direct-joins process rows or blocks in random order until the budget is spent.

\textbf{Observation \theexp. With \$1, \system{} exceeds what the LLM-direct-joins and the semantic operators reach with an unlimited budget.}
\system{} reaches 0.922 with \$1 and 0.922 to 0.942 with larger budgets. The LLM-direct-joins grow roughly in proportion to the budget, to 0.186 at \$16, and the semantic operators stay at or below 0.010. With an unlimited budget, these baselines reach at most 0.770 and 0.771, at estimated costs of \$588 and \$1.0 million per dataset, while the coding agents spend less than \$1 and vary without a trend.
With an unlimited budget, it spends \$4.44 per dataset on average, so a larger budget changes its gated recall little.

\begin{table}[t]
\small
  \caption{(Exp.~4) Gated recall $GR_\tau$ of \system{} and the LLM-based baselines on \joinhard{} under different budgets, and API cost with an unlimited budget. $\dagger$ marks estimates from sampled source rows, and parentheses give the semantic operators' cost lower bounds.}
  \label{tab:varying_budget}
  \centering
  \setlength{\tabcolsep}{3.5pt}
  \resizebox{\columnwidth}{!}{%
  \begin{tabular}{llcccccccr}
    \toprule
    \multicolumn{2}{l}{Budget $B$} & \$0.5 & \$1 & \$2 & \$4 & \$8 & \$16 & \multicolumn{2}{c}{Unlimited budget} \\
    \cmidrule(lr){3-8} \cmidrule(lr){9-10}
    \multicolumn{2}{l}{Method} & $GR_\tau$ & $GR_\tau$ & $GR_\tau$ & $GR_\tau$ & $GR_\tau$ & $GR_\tau$ & $GR_\tau$ & Cost (\$) \\
    \midrule
    \multirow{4}{*}{\makecell[l]{LLM semantic\\operator}}
     & LOTUS & 0.010 & 0.010 & 0.010 & 0.010 & 0.010 & 0.010 & 0.552$^\dagger$ & 1{,}126{,}229$^\dagger$ (3{,}108) \\
     & FDJ & 0.000 & 0.000 & 0.000 & 0.000 & 0.000 & 0.000 & 0.577$^\dagger$ & 1{,}982{,}447$^\dagger$ (489) \\
     & AISQL & 0.000 & 0.000 & 0.000 & 0.000 & 0.000 & 0.000 & \rkthird{0.771$^\dagger$} & 1{,}028{,}671$^\dagger$ (301) \\
     & Adaptive-Join & 0.000 & 0.000 & 0.000 & 0.000 & 0.000 & 0.000 & 0.155$^\dagger$ & 396{,}434$^\dagger$ (3{,}100) \\
    \midrule
    \multirow{2}{*}{Coding agent}
     & Codex-CLI & \rksecond{0.663} & \rksecond{0.771} & \rksecond{0.756} & \rksecond{0.766} & \rksecond{0.754} & \rksecond{0.778} & 0.682 & 0.44 \\
     & GitHub-Copilot-CLI & \rkthird{0.635} & \rkthird{0.698} & \rkthird{0.655} & \rkthird{0.707} & \rkthird{0.642} & \rkthird{0.672} & \rksecond{0.797} & 0.48 \\
    \midrule
    \multirow{3}{*}{LLM-direct-join}
     & Embedding top-1 & 0.004 & 0.013 & 0.020 & 0.043 & 0.103 & 0.186 & 0.649$^\dagger$ & 108$^\dagger$ \\
     & Embedding top-10 & 0.001 & 0.003 & 0.004 & 0.007 & 0.018 & 0.037 & 0.770$^\dagger$ & 588$^\dagger$ \\
     & Block nested loop join & 0.000 & 0.000 & 0.000 & 0.001 & 0.002 & 0.003 & 0.355$^\dagger$ & 97{,}310$^\dagger$ \\
    \specialrule{\lightrulewidth}{\aboverulesep}{0pt}
    \rowcolor{gray!15}
     & \textbf{\system{}}\sysrowstrut & \rkfirst{\textbf{0.862}} & \rkfirst{\textbf{0.922}} & \rkfirst{\textbf{0.922}} & \rkfirst{\textbf{0.923}} & \rkfirst{\textbf{0.923}} & \rkfirst{\textbf{0.924}} & \rkfirst{\textbf{0.942}} & \textbf{4.44} \\
    \specialrule{\heavyrulewidth}{0pt}{\belowbottomsep}
  \end{tabular}}
\end{table}

\stepcounter{exp}
\stitle{Experiment \theexp: Varying precision target.}
Figure~\ref{fig:tau_precision} shows, as box plots over the datasets of \joinhard{} and AFJ, the precision of \system{} under precision targets $\tau$ from 0.80 to 0.98 with $B = \$1$.

\textbf{Observation \theexp. The precision guard meets every target up to 0.90 on every dataset; above that, the precision is bounded by the precision of the LLM-oracle.}
For $\tau$ from 0.80 to 0.90, every result meets $\tau$ on both benchmarks, and the precision rises with $\tau$ (on \joinhard{}, the median grows from 0.955 to 0.979), because a higher target admits fewer batches of lower precision. At $\tau = 0.95$ and $0.98$, the mean precision meets $\tau$ except on AFJ at $\tau = 0.98$ (0.972), but some datasets fall below $\tau$: 12 and 27 of the 50 AFJ datasets, and 1 of the 10 \joinhard{} datasets at $\tau = 0.98$. Proposition~\ref{prop:gate} assumes that the verdicts of the LLM-oracle are correct, and these targets approach or exceed the precision of the LLM-oracle (0.962 on AFJ, Table~\ref{tab:oracle_precision}).

\begin{figure}[t]
\centering
\includegraphics[width=\columnwidth]{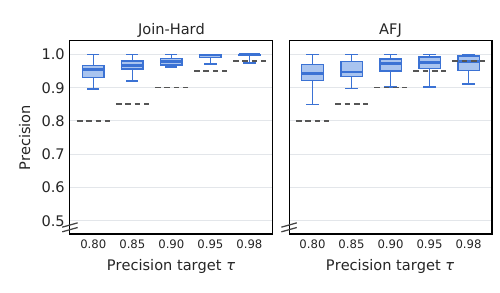}
\caption{(Exp.~5) Precision of \system{} on \joinhard{} and AFJ under five precision targets $\tau$. Boxes span the 25th to 75th percentiles, whiskers the minimum and maximum, and dashed lines mark $\tau$. }
\label{fig:tau_precision}
\end{figure}

\stepcounter{exp}
\stitle{Experiment \theexp: Ablation of join types.}
Table~\ref{tab:ablation_arms} reports the gated recall, precision, recall, and $F_1$ of \system{} with \texttt{gpt-5.4} after removing join types from the pool, averaged over the 141 datasets of the four benchmarks.
\emph{w/o Fuzzy} removes both J1 arms but keeps the fuzzy step of the compositions; \emph{w/o Transform} and \emph{w/o Semantic} remove the J2 or the J3 arms with their compositions; and \emph{w/o LLM-direct-join} removes the 12 J4 arms.

\textbf{Observation \theexp. Every join type contributes: removing any one lowers the gated recall, and no type is fully replaced by the others.}
Removing any one type lowers the gated recall from 0.895 to at most 0.872, while the precision stays at least 0.968.

\begin{table}[t]
\small
  \caption{(Exp.~6) Gated recall, precision, recall, and $F_1$ of \system{} without one join type, averaged over 141 datasets; bold and underline mark the best and second-best per column.}
  \label{tab:ablation_arms}
  \centering
  \setlength{\tabcolsep}{3.5pt}
\begin{tabular}{lcccc}
    \toprule
    Variant & $GR_\tau$ & $P$ & $R$ & $F_1$ \\
    \midrule
    w/o Fuzzy & 0.861 & 0.968 & 0.874 & 0.906 \\
    w/o Transform & 0.847 & 0.972 & 0.861 & 0.890 \\
    w/o Semantic & \underline{0.872} & 0.971 & \underline{0.891} & \underline{0.920} \\
    w/o LLM-direct-join & 0.850 & \textbf{0.975} & 0.855 & 0.897 \\
    \specialrule{\lightrulewidth}{\aboverulesep}{0pt}
    \rowcolor{gray!15}
    \textbf{\system{}}\sysrowstrut & \textbf{0.895} & \underline{0.974} & \textbf{0.895} & \textbf{0.925} \\
    \specialrule{\heavyrulewidth}{0pt}{\belowbottomsep}
  \end{tabular}
\end{table}

\section{Related Work}
\label{sec:related}

\stitle{Program-based joins.}
Automatic join methods handle keys that differ, each for one join type: fuzzy-joins use set-similarity algorithms~\cite{ssjoin, ppjoin, jiang2014strjoin} or label-free matchers~\cite{autofuzzyjoin, zeroem}, transformation-joins search for programs of string operators~\cite{autojoin, cst, efficient-auto-join-2, omidvartehrani2024gxjoin}, and semantic-joins mine mappings from table corpora~\cite{semajoin, mapping-tables}.
Auto-FuzzyJoin~\cite{autofuzzyjoin} also maximizes recall under a precision target, but only for fuzzy-joins, and it estimates precision without labels by assuming that the reference table has no duplicates.
Programming by example synthesizes data transformations from user-provided input-output examples~\cite{flashfill, blinkfill, foofah, tde, dataxformer}, and LLMs write them as code~\cite{nobari2024tabulax} or generate the transformed values directly~\cite{dtt}.
\system{} creates instances of several join types with such methods or an LLM and guarantees the precision of their combined result.

\stitle{LLM-based joins.}
Semantic operators in LLM-powered query processing include joins whose condition is stated in natural language~\cite{lotus, thalamusdb, palimpzest, docetl, abacus}, and LLM-based entity matching decides whether two records match~\cite{narayan2022wrangle, matchgpt, comem, batcher}.
Semantic-operator joins reduce LLM calls with embeddings, extracted features, batching, or classification~\cite{lotus, fdj, adaptive-join, aisql}, but their cost still grows with the rows that the LLM reads.
Coding agents can also write a join script from samples of the two tables, without a guarantee on the precision of its output; \S\ref{sec:exp-results} compares three of them with \system{}.
\system{} has an LLM read rows only in LLM-direct-joins (J4) on the rows not yet joined, and otherwise uses it to create instances and judge samples of their outputs.

\stitle{Budgeted optimization and multi-armed bandits.}
Multi-armed bandits balance exploration and exploitation over arms with fixed reward distributions~\cite{bubeck2012regret}; variants add budgets~\cite{tranthanh2010epsilon, tranthanh2012budget, ding2013budget, badanidiyuru2018bwk}, arms that arrive, expire, or come from a reservoir~\cite{whittle1981arm, chakrabarti2008mortal, li2017infinite, kalvit2021arm}, decaying rewards~\cite{levine2017rotting}, or paid observations~\cite{seldin2014paid}.
A cost-aware greedy rule approximates budgeted maximum coverage~\cite{khuller1999budgeted}, also in adaptive settings~\cite{golovin2011adaptive, gabillon2013adaptive}.
\system{} borrows exploration, exploitation, and the gain-per-dollar rule, but its arms have two levels, and the rewards of its instances overlap, diminish as $J$ grows, and are observed only through paid LLM-oracle samples.
Because its committed outputs must also jointly meet a precision target, the guarantees of these works do not carry over.

\section{Conclusion}
\label{sec:conclusion}

This paper studied Precision-Guaranteed Budget-Constrained Join, which joins two tables with several join types under a dollar budget while keeping the precision of the result at or above a target.
Our key idea is to use LLMs selectively to create join programs and verify their outputs through sampling, while adaptively deciding which joins to pursue under budget and precision constraints.
Under the same budget, \system{} achieves the highest gated recall in nearly all settings on four benchmarks, including \joinhard{}.
Future work includes relaxing the assumption that the verdicts of the LLM-oracle are correct and adding join types to the pool.

\clearpage
\balance
\bibliographystyle{ACM-Reference-Format}
\bibliography{ref}

\clearpage
\nobalance
\appendix
\section{Proof of Theorem 1}

Join instance selection (\S3.2) takes a finite set of instances $\mathcal{X}$, the join result $J_x$ and cost $r_x$ of each $x \in \mathcal{X}$, the correct pairs $J^{*}$, a precision target $\tau$, and a budget $B$, and selects $X \subseteq \mathcal{X}$ that solves
\[
\max_{X \subseteq \mathcal{X}}\ \text{recall}(J_X) \quad \text{s.t.} \quad \text{precision}(J_X) \geq \tau, \quad \sum_{x \in X} r_x \leq B.
\]
Theorem~1 states that join instance selection is \textsf{NP}-hard.

\begin{proof}
The reduction is from Budgeted Maximum Coverage (BMC)~\cite{khuller1999budgeted}, which is \textsf{NP}-hard. An instance of BMC consists of a universe $E = \{e_1, \dots, e_n\}$, subsets $F_1, \dots, F_K \subseteq E$ with costs $w_1, \dots, w_K$, and a budget $W$. The task is to select subsets of total cost at most $W$ that together cover the most elements of $E$.

Given such an instance, construct an instance of join instance selection in polynomial time. Let $S = \{s_1, \dots, s_n\}$, $T = \{t_1, \dots, t_n\}$, and $J^{*} = \{(s_j, t_j) : 1 \leq j \leq n\}$. For each subset $F_i$, let $\mathcal{X}$ contain an instance $x_i$ with $J_{x_i} = \{(s_j, t_j) : e_j \in F_i\}$ and $r_{x_i} = w_i$, and let $B = W$. Every pair of every instance is in $J^{*}$, so $\text{precision}(J_X) = 1 \geq \tau$ for every nonempty $X \subseteq \mathcal{X}$. For $X = \{x_i : i \in I\}$,
\[
\text{recall}(J_X) = \frac{1}{n}\Bigl|\bigcup_{i \in I} F_i\Bigr|,
\qquad
\textstyle\sum_{x \in X} r_x = \sum_{i \in I} w_i,
\]
so $X$ is an optimal solution of join instance selection exactly when $\{F_i : i \in I\}$ is an optimal solution of BMC. A polynomial-time algorithm for join instance selection would therefore solve BMC, which is \textsf{NP}-hard.
\end{proof}

\section{Proof of Proposition 1}

\stitle{Precision guard.}
The proof uses the following properties of the precision guard (\S4.3 and \S4.5).
The $k$-th batch of a run, $Q_k$, has $N_k$ pairs, of which $C_k$ are correct.
It contains only pairs whose rows of $S$ are not yet joined when it is first sampled, it does not change while it is checked, and no other batch is committed between its checks.
Its first check judges $\min(n, N_k)$ pairs of $Q_k$, drawn uniformly without replacement, and each of its at most two further checks judges $\min(n, N_k - m)$ more pairs, drawn uniformly without replacement from the pairs not yet judged, where $m$ is the number of pairs judged so far.
A check whose LLM-oracle response the cap on output tokens cuts short counts as not performed (\S4.5), so every check that the run performs judges all the pairs drawn for it.
If $c$ of the $m$ judged pairs are judged correct, the check uses the bound $C^{-}_k = c$ if $m = N_k$, and otherwise
\[
\begin{aligned}
C^{-}_k &= \min\{C' \in \{0, \dots, N_k\} : \Pr[Y \geq c] > \alpha_k\},\\
Y &\sim \mathrm{Hypergeometric}(N_k, C', m),
\end{aligned}
\]
where $\alpha_k = \alpha / (3 \cdot 2^j)$ if batch $k$ is the $j$-th partially sampled batch of the run, that is, the $j$-th batch whose first check judges fewer than $N_k$ pairs; the other batches are fully judged by their first check.
The batch would add $A_k = N_k - (m - c)$ pairs, since its pairs judged incorrect are left out, and a batch with $A_k = 0$ is skipped.
The precision guard commits batch $k$ if
\[
\textstyle\sum_i C^{-}_i + C^{-}_k \geq \tau \bigl(\sum_i A_i + A_k\bigr),
\]
where $i$ ranges over the records of the batches in $J$, which is equivalent to $L \geq \tau$ in \S4.5, and it then records $(C^{-}_k, A_k)$.
If it rejects the batch after its last check, it commits the pairs of its samples judged correct, if any, with a record whose $C^{-}$ and $A$ both equal their number.
An empty join result is taken to have precision 1.

\stitle{Notation.}
By the assumption of Proposition~1, the LLM-oracle judges a pair correct if and only if the pair belongs to $J^{*}$.
Let $\mathcal{F}_k$ denote the history of the run before the first sample of batch $k$, that is, all actions, samples, and verdicts up to that point.
The batch $Q_k$, and whether it is the $j$-th partially sampled batch, are determined by $\mathcal{F}_k$.
Since its samples are drawn uniformly without replacement, conditional on $\mathcal{F}_k$, the pairs of $Q_k$ are judged in the order of a uniformly random permutation $\pi_k$ of $Q_k$, so the first $m$ pairs judged are the first $m$ entries of $\pi_k$.

\begin{lemma}\label{lem:hyper}
Let $0 \leq m \leq N$, $0 \leq C \leq N$, and $0 < a < 1$. For an integer $0 \leq c \leq m$, let $g_c(C') = \Pr[X' \geq c]$ with $X' \sim \mathrm{Hypergeometric}(N, C', m)$, and let $C^{-}(c) = \min\{C' : g_c(C') > a\}$. Then $c \leq C^{-}(c) \leq N - m + c$, with $C^{-}(c) = c$ if $m = N$, and $\Pr[C^{-}(X) > C] \leq a$ for $X \sim \mathrm{Hypergeometric}(N, C, m)$.
\end{lemma}
\begin{proof}
A population with $C'+1$ correct pairs is obtained from one with $C'$ correct pairs by marking one incorrect pair as correct, which does not decrease the number of correct pairs in any sample. Hence $g_c(C')$ is nondecreasing in $C'$. Since $g_c(C') = 0$ for $C' < c$, and $g_c(N-m+c) = 1$ because a population with $N-m+c$ correct pairs has only $m-c$ incorrect pairs, $C^{-}(c)$ exists and $c \leq C^{-}(c) \leq N-m+c$. If $m = N$, then $X' = C'$, so $g_c(C') > a$ holds exactly when $C' \geq c$, and $C^{-}(c) = c$.

By the monotonicity of $g_c$, $C^{-}(c) > C$ holds if and only if $g_c(C) \leq a$. The function $c \mapsto g_c(C) = \Pr[X \geq c]$ is nonincreasing, so the set of $c$ with $g_c(C) \leq a$ is either empty or equal to $\{c : c \geq c_0\}$ for some $c_0$. In the first case, $\Pr[C^{-}(X) > C] = 0$. In the second case, $\Pr[C^{-}(X) > C] = \Pr[X \geq c_0] = g_{c_0}(C) \leq a$.
\end{proof}

\begin{lemma}\label{lem:batch}
Let batch $k$ be the $j$-th partially sampled batch of the run. For $\ell \in \{1, 2, 3\}$, let $m_{k,\ell} = \min(\ell n, N_k)$, let $X_{k,\ell}$ be the number of pairs of $J^{*}$ among the first $m_{k,\ell}$ entries of $\pi_k$, and let $C^{-}_{k,\ell}$ be the bound $C^{-}(X_{k,\ell})$ of Lemma~\ref{lem:hyper} with $N = N_k$, $m = m_{k,\ell}$, and $a = \alpha_k$. Then the $\ell$-th check of batch $k$, if the run performs it, uses $C^{-}_{k,\ell}$, and
\[
\Pr\bigl[\, C^{-}_{k,\ell} > C_k \text{ for some } \ell \in \{1,2,3\} \bigm| \mathcal{F}_k \bigr] \leq \frac{\alpha}{2^j}.
\]
\end{lemma}
\begin{proof}
Before its $\ell$-th check, the run has judged $m_{k,\ell}$ pairs of batch $k$, which are the first $m_{k,\ell}$ entries of $\pi_k$, and $X_{k,\ell}$ of them are judged correct. If $m_{k,\ell} < N_k$, the check uses $C^{-}_{k,\ell}$ by definition. If $m_{k,\ell} = N_k$, the check uses $X_{k,\ell}$, which equals $C^{-}_{k,\ell}$ by Lemma~\ref{lem:hyper}.

Conditional on $\mathcal{F}_k$, the first $m_{k,\ell}$ entries of $\pi_k$ form a uniform sample without replacement from $Q_k$, so $X_{k,\ell} \sim \mathrm{Hypergeometric}(N_k, C_k, m_{k,\ell})$. Lemma~\ref{lem:hyper} then gives $\Pr[C^{-}_{k,\ell} > C_k \mid \mathcal{F}_k] \leq \alpha_k$ for each $\ell$. The variables $C^{-}_{k,\ell}$ are defined for all three values of $\ell$, whether or not the run performs the $\ell$-th check, so the decision to draw more pairs does not affect this bound. The union bound over $\ell$ gives $3 \alpha_k = \alpha / 2^j$.
\end{proof}

\begin{lemma}\label{lem:all}
Let $E$ be the event that $C^{-}_{k,\ell} \leq C_k$ for every partially sampled batch $k$ of the run and every $\ell \in \{1,2,3\}$. Then $\Pr[E] \geq 1 - \alpha$.
\end{lemma}
\begin{proof}
For $j \geq 1$, let $D_j$ be the event that the run has a $j$-th partially sampled batch $k$ and that $C^{-}_{k,\ell} > C_k$ for some $\ell$. Whether this batch exists, and which batch it is, is determined by the history before its first sample. Conditioning on this history, Lemma~\ref{lem:batch} gives $\Pr[D_j] \leq \alpha / 2^j$. The event $E$ fails exactly when some $D_j$ occurs, so $\Pr[\neg E] \leq \sum_{j \geq 1} \alpha / 2^j = \alpha$.
\end{proof}

\begin{lemma}\label{lem:records}
For each record $i$ at the end of the run, let $P_i$ be the set of pairs that the record adds to $J$. The sets $P_i$ are pairwise disjoint, their union is $J$, and $|P_i| = A_i$. On $E$, $|P_i \cap J^{*}| \geq C^{-}_i$ for every record $i$.
\end{lemma}
\begin{proof}
Every pair of a batch joins a row of $S$ that is not yet joined at the first sample of the batch, no other batch is committed between its checks, and committed pairs join their rows of $S$. Pairs of different records therefore join different rows of $S$, so the sets $P_i$ are pairwise disjoint. $J$, the set of committed pairs, is the union of the sets $P_i$. For the second claim, records are of three kinds.
\begin{itemize}
\item A committed batch $k$ that is fully judged by its first check adds its pairs judged correct. These pairs belong to $J^{*}$, so $|P_i \cap J^{*}| = c_k = C^{-}_i$.
\item A committed batch $k$ that is partially sampled, whose last check is the $\ell$-th (which may judge all its pairs), adds $Q_k$ without its pairs judged incorrect. These pairs do not belong to $J^{*}$, so $|P_i \cap J^{*}| = C_k$, and $C_k \geq C^{-}_{k,\ell} = C^{-}_i$ on $E$.
\item The pairs committed from a rejected batch are judged correct, so they belong to $J^{*}$, and $|P_i \cap J^{*}| = A_i = C^{-}_i$. \qedhere
\end{itemize}
\end{proof}

\begin{proof}[Proof of Proposition~1]
Let $\mathcal{C}_t$ and $\mathcal{A}_t$ be the sums of $C^{-}_i$ and of $A_i$ over the records after the $t$-th commit of the run, with $\mathcal{C}_0 = \mathcal{A}_0 = 0$. By induction on $t$, $\mathcal{C}_t \geq \tau \mathcal{A}_t$. When the precision guard commits a batch $k$, it has checked $\mathcal{C}_{t-1} + C^{-}_k \geq \tau (\mathcal{A}_{t-1} + A_k)$. When it commits $a$ pairs from a rejected batch, $\mathcal{C}_{t-1} + a \geq \tau \mathcal{A}_{t-1} + \tau a$, because $\tau \leq 1$. Let $T$ be the number of commits of the run. By Lemma~\ref{lem:records}, $|J| = \mathcal{A}_T$, and on $E$,
\[
|J \cap J^{*}| = \sum_i |P_i \cap J^{*}| \geq \mathcal{C}_T \geq \tau \mathcal{A}_T = \tau |J|.
\]
On $E$, $\text{precision}(J) \geq \tau$ therefore holds if $J$ is nonempty, and by convention if $J$ is empty. Lemma~\ref{lem:all} gives $\Pr[\text{precision}(J) \geq \tau] \geq \Pr[E] \geq 1 - \alpha$. Lemmas~\ref{lem:batch} and~\ref{lem:all} condition on the full history before the first sample of each batch, which contains every choice of the planner, so this bound holds for any rule that selects the next batch from the outcomes of earlier rounds.
\end{proof}

\end{document}